\documentclass[12pt,reqno]{amsart}
\usepackage{graphicx}
\usepackage{amssymb,amsmath}
\usepackage{amsthm}
\usepackage{color}
\usepackage[pdf]{pstricks}
\usepackage{hyperref}
\usepackage{empheq}
\usepackage{pgfplots}

\usepackage{mathtools}

\DeclareMathOperator{\sech}{sech}

\usepackage[normalem]{ulem}

\usepackage{tikz}
\usepackage{graphicx}  
\usepackage{amsmath}  
\usepackage{amssymb,amsthm}
\usepackage{color}

\newcommand{\R}{{\mathbb R}}
\newcommand{\C}{{\mathbb C}}

\newtheorem{theorem}{Theorem}
\newtheorem{lemma}{Lemma}

\newtheorem{remark}{Remark}

\begin{document}

\title[Persistence of kinks]{\bf Persistence of kinks \\ in the NLS equation with competing nonlinearities}

\author[P. G. Kevrekidis]{Panayotis G. Kevrekidis}
\address[P. G. Kevrekidis]{Department of Mathematics and Statistics, University of Massachusetts Amherst, Amherst, MA 01003-4515, USA}
\address[P. G. Kevrekidis]{Department of Physics, University of Massachusetts Amherst, Amherst, MA 01003, USA}
\email{kevrekid@umass.edu}

\author[D. E. Pelinovsky]{Dmitry E. Pelinovsky}
\address[D. E. Pelinovsky]{Department of Mathematics and Statistics, McMaster University, Hamilton, Ontario, Canada, L8S 4K1}
\email{pelinod@mcmaster.ca}

\begin{abstract}
External potentials break the translational invariance of nonlinear Schr\"odinger (NLS) equations and generally convert a continuous family of standing wave solutions into isolated pinned states.  We study this mechanism for the prototypical case of a kink in the NLS equation with competing focusing cubic and defocusing quintic nonlinearities. The Lyapunov--Schmidt reduction method is used to prove that the location $s_0$ of the pinned kinks is selected near roots of an effective potential $V_{\rm eff}(s)$. We then determine how the translational zero eigenvalue of the Jacobian operator splits under the perturbation and prove that the sign of $V_{\rm eff}'(s_0)$ distinguishes  eigendirections of positive and negative energy. We further show that if the Jacobian operator admits a negative eigenvalue, then the stability problem contains a real unstable eigenvalue.
The derivation of the asymptotic expansion for the unstable eigenvalue is nonstandard because the zero eigenvalue is embedded in the continuous spectrum of the linearized operator. Finally, we show that the persistence of kinks lacks parity symmetry even if the external potential is odd. Numerical continuation, spectral computations with domains adapted to the slowly decaying eigenfunctions, and direct dynamical
simulations of the time evolution corroborate the analytical predictions for different representative examples.
\end{abstract}

\maketitle

\section{Introduction}
\label{sec:introduction}

\subsection{Background} 

Nonlinear Schr\"odinger (NLS) models with competing nonlinearities arise naturally when distinct physical mechanisms dominate at different amplitudes.  A canonical example is the focusing--defocusing cubic--quintic NLS equation, which has long served as a prototypical model for saturation and higher-order nonlinear response in optics and for the coexistence of qualitatively different coherent structures \cite{CrasovanMalomedMihalache2001}.  In nonlinear optical media, the simultaneous action of third- and fifth-order susceptibilities has been observed experimentally and modeled successfully by cubic--quintic NLS-type equations \cite{FalcaoFilhoEtAl2013,ReynaEtAl2020}.  The same general competition between attractive and repulsive effects has also become central in ultracold gases \cite{Petrov2015,PetrovAstrakharchik2016}, where the beyond-mean-field corrections (e.g., due to quantum fluctuations)  can arrest the mean-field collapse and produce self-bound quantum droplets \cite{LuoEtAl2021}. The mathematical analysis of the cubic--quintic NLS equation, including the problems of uniqueness, nondegeneracy, scattering, and orbital stability, has also received considerable attention, see  \cite{KillipOhPocovnicuVisan2017}, \cite{LewinRotaNodari2020}, and \cite{CarlesSparber2021}.

Competing nonlinearities support a particularly rich set of nonlinear states.  In the one-dimensional setting, bright droplets, bubbles, black solitons, and  kinks can occur in neighboring parameter regimes \cite{KatsimigaEtAl2023PRA}.  The kink is distinguished by connecting two different asymptotic states, one vanishing and the other nonzero.  In this sense it shares features with both bright and black solitons: one side approaches the zero background, whereas the other carries a finite-amplitude background.  
It is also rather special, as it emerges at a threshold between other states (such as droplets and bubbles) and in that
sense, it arises for special values of the system parameters.
Recent work has shown that such kinks can be remarkably 
stable against transverse perturbations in higher-dimensional settings with competing attractive and repulsive interactions \cite{MistakidisEtAl2025PRL}, compared to the transverse instability of bright solitons in the cubic--quintic NLS equation \cite{KleinSparber2026}.

The stability of kinks and black solitons is mathematically subtle because nonzero boundary conditions make it hard to control phase modulations outside the localized state by using the energy functional \cite[Chapter 7]{GeyerPelinovsky2025}.  For the defocusing cubic NLS equation, orbital stability of the black solitons was established in the natural energy space where the  phase modulations were not properly controlled \cite{BethuelGravejatSautSmets2008}. The exponentially weighted energy space was introduced in \cite{GravejatSmets2015} and used in \cite{GallayPelinovsky2015} to control the phase modulations outside the black solitons.  For the kink of the cubic--quintic NLS equation considered here, a recent rigorous analysis was developed in \cite{HolmerKevrekidisPelinovsky2025} by using the exponentially weighted energy space. The main difference between the black solitons and kinks is that the black solitons are constrained minimizers of energy subject to fixed momentum, whereas the kinks are unconstrained minimizers of energy. These stability and energy minimization results provide a natural starting point for asking what becomes of the kink when the translational invariance is broken due to an external potential. This question is similar to the
effective dynamics of bright and dark solitons in external potentials, see~\cite{BronskiJerrard2000,Sigal2004}, \cite{KevPel08}, and \cite{Bengel2026,BPR24}.

\subsection{Main results}
 
{\em The purpose of the present work is to answer the persistence question for kinks in a small  potential $\varepsilon V$ with $0 < \varepsilon \ll 1$.}  In the unperturbed cubic--quintic problem, the kink belongs to the translational orbit
\[
\{\phi_0(\cdot-s):s\in\mathbb{R}\}.
\]
Consequently, the Jacobian operator at $\phi_0$ possesses a zero eigenvalue generated by the translational mode $\phi_0'$.  Once an external potential 
$V(x) : \mathbb{R}\to \mathbb{R}$ is introduced, the translation parameter $s$ is no longer free and only pinned kinks can generally persist  for isolated values of $s$.  At the same time, the zero eigenvalue of 
the Jacobian operator moves away from the origin, and its sign determines whether the pinned state remains a local energy minimizer or becomes a saddle point of energy.  

The corresponding spectral stability problem has an additional complication.  Because the kink approaches a nonzero state at one spatial infinity, the zero eigenvalue is no longer isolated, but is embedded into the continuous spectrum of the linearized operator. The standard feature of the linearized Hamiltonian system for bright solitons given by the square-root splitting of a double isolated eigenvalue (see, e.g., \cite{Kapitula}) does not apply in the stability problem for kinks. 
Furthermore, the splitting of the embedded zero eigenvalue for kinks is very different from the splitting of the embedded zero eigenvalue for black solitons \cite{KevPel08}. 

\medskip

{\em Our analytical part contains four results. } 

\medskip

\underline{First,} the Lyapunov--Schmidt reduction method is used to convert the persistence problem of the kink orbit under a small potential $\varepsilon V$ into a root-finding problem involving the effective potential $V_{\rm eff}(s) : \R \to \R$ given by 
\[
V_{\rm eff}(s)=\frac14\int_{\mathbb{R}}V(x)\sech^2(x-s)\,dx.
\]
A simple zero $s_0$ of $V_{\rm eff}$ selects a unique pinned kink near $\phi_0(\cdot-s_0)$ for sufficiently small perturbation strength $0 < \varepsilon \ll 1$. The result is proven in Theorem \ref{th-existence} under the only condition $V \in L^2(\R)$, which is used to close the fixed-point arguments of the Lyapunov--Schmidt reduction method.

\medskip

\underline{Second,} the perturbation theory for the Jacobian operator relates the shift of the small eigenvalue associated with the translational mode $\phi'_0$ to the sign of $V_{\rm eff}'(s_0)$ for the simple zero $s_0$. The precise asymptotic behavior of the eigenvalue is given by 
$$
\mu_+(\varepsilon) = -4 \varepsilon V'_{\rm eff}(s_0) + \mathcal{O}(\varepsilon^2).
$$  
In particular, $V_{\rm eff}'(s_0)<0$ induces a simple positive eigenvalue, whereas $V_{\rm eff}'(s_0)>0$ induces a simple negative eigenvalue.  This result is given by Theorem \ref{th-splitting} under the condition $V \in H^1(\R)$, 
where the additional constraint $V' \in L^2(\R)$ is needed to justify the result of the first-order perturbation theory and to relate projections involving $V'$ with the derivative of $V_{\rm eff}$ at $s_0$. If we add the variational interpretation of the result of Theorem \ref{th-splitting}, then 
we additionally need to restrict $V \in L^1(\R)$ in order to define the energy functional in the energy space. Under this restriction $V \in H^1(\R) \cap L^1(\R)$, the kink is a local energy minimizer if $V_{\rm eff}'(s_0)<0$ and an energy saddle point if $V_{\rm eff}'(s_0) > 0$.

\medskip

\underline{Third,} we analyze the spectral stability problem and obtain a real unstable eigenvalue if $V_{\rm eff}'(s_0) > 0$ for the simple zero $s_0$. The precise asymptotic behavior of the eigenvalue is given by 
\[
\lambda(\varepsilon)=2\varepsilon V_{\rm eff}'(s_0)+\mathcal{O}(\varepsilon^2).
\]
The linear dependence of the small eigenvalue on $\varepsilon$, rather than the standard splitting $\varepsilon^{1/2}$ associated with an isolated double zero eigenvalue in the linearized Hamiltonian systems, reflects the embedding property of the zero eigenvalue in the unperturbed spectral problem. This result is proven in Theorem \ref{th-stability} under the condition $V \in H^1(\R) \cap L^1(\R)$. The restriction of $V \in L^1(\R)$ provides precise information about the spatial decay of eigenfunctions for small eigenvalues $\lambda$ of the spectral stability problem with ${\rm Re}(\lambda) > 0$.

\medskip

\underline{Finally,} we address the parity symmetry of the kink solutions. In the unperturbed problem, the kink is symmetric relative to the level $\frac{1}{\sqrt{2}}$ in the sense that $v_0 := 2 \phi_0^2 - 1$ is an odd function of $x$ on $\mathbb{R}$. We show in Theorem \ref{th-parity} that in the perturbed problem with the kink profile $\phi_{\varepsilon}$, the function $v_{\varepsilon} := 2 \phi_{\varepsilon}^2 - 1$ is not odd even if $V$ is odd in $x$. The result is obtained for $V \in H^1(\R)$ and we use the continuity of $V$ from Sobolev embedding of $H^1(\R)$ to $C^0(\R) \cap L^{\infty}(\R)$ to prove that the even terms of the differential equation for $v_{\varepsilon}$ are nonzero even if $V$ is odd. The result of Theorem \ref{th-parity} is rather surprising since $V_{\rm eff}$ is odd in $s$ if $V$ is odd in $x$, which implies that the odd potential $V$ supports pinning at $s_0 = 0$ at the leading order in $\varepsilon$. Theorem \ref{th-parity} suggests that the pinning location $s(\varepsilon)$ is nonzero beyond the leading order in $\varepsilon$.

\medskip

{\em The analytical predictions are complemented by direct numerical computations.}  

\medskip

We first examine odd potentials $V(x) = \pm x e^{-x^2}$, for which the selected center is $s_0=0$ and the sign of the potential reverses the energetic and spectral conclusions.  We then study the continuous family of even potentials
\[
V(x) = (1-a x^2)\sech^2(x), \quad a > 0,
\]
for which two symmetry-related pinned branches coexist and merge at a finite value of $a > 0$.  Pseudo-arclength continuation captures the fold bifurcation diagram, while computations of the Jacobian operator and of the spectral stability spectrum test the predicted eigenvalue asymptotics.  Since the eigenfunction for the unstable eigenvalue develops a slowly decaying tail as the eigenvalue approaches the origin, the spectral computations use domains adapted to that decay scale. This is essential for separating genuine spectral behavior from finite-domain effects.  

Finally, time integration illustrates the dynamical departure of the perturbed kink from the unstable branch and the ensuing motion toward the stable branch, according to the effective viscous particle equation 
$$
\frac{ds}{dt} = 2 \varepsilon V_{\rm eff}(s), \quad s \in \mathbb{R}.
$$
The justification of the validity of the viscous particle equation for all times $t \geq 0$ and the asymptotic stability of the pinned kink as an energy minimizer are left open for further studies. We point out that the viscous particle equation was justified in the recent work \cite{Bengel2026} for the damped forced NLS equation (the Lugiato--Lefever model), where it was also shown how well the 1D dynamics predicts the numerical evolution of the pinned solutions. Compared to the pinned kink in the NLS equation, the 1D dynamics in the Lugiato--Lefever equation is not surprising since the bifurcating simple eigenvalue is isolated from the continuous spectrum of the linearized operator located in the left half of the complex plane \cite{BPR24}.

\subsection{Organization of the paper}  

Section \ref{sec-2} establishes persistence of the kink by the  Lyapunov--Schmidt reduction method and identifies the effective potential $V_{\rm eff}$.  Section \ref{sec-3} studies the second variation of the energy and the splitting of the translational eigenvalue $\mu_+(\varepsilon)$ of the Jacobian operator.  Section \ref{sec-4} analyzes the spectral stability problem and derives the asymptotic expansion for the  unstable eigenvalue $\lambda(\varepsilon)$. Section \ref{sec-5} addresses the lack of parity symmetry for the pinned kink. Section \ref{sec:numerics} compares the analytical predictions with the numerical continuation, spectral calculations, and nonlinear dynamics of kinks.

\section{Persistence argument}
\label{sec-2}

We consider the cubic--quintic NLS equation with the external potential,
\begin{equation}
\label{NLS}
i \psi_t - \psi_{xx} + \varepsilon V(x) \psi + \psi - 4 |\psi|^2 \psi + 3 |\psi|^4 \psi = 0, \quad x \in \mathbb{R}, 
\end{equation}
where $\varepsilon$ is a small parameter and $V(x) : \R \to \R$ is a fixed potential under some suitable assumptions. We consider time-independent real-valued solutions of the stationary equation 
\begin{equation}
\label{ode}
\phi'' = \phi (1 - \phi^2) (1- 3 \phi^2) + \varepsilon V \phi, \quad x \in \mathbb{R}, 
\end{equation}
subject to the boundary conditions $\phi(x) \to 0$ as $x \to -\infty$ and $\phi(x) \to 1$ as $x \to +\infty$. These solutions are referred to as {\em kinks}. 

If $\varepsilon = 0$, then the kink solution exists with the monotonically increasing profile $\phi$ found from the first-order reduction
\begin{equation}
\label{first-order}
\varepsilon = 0 : \quad \phi' = \phi (1-\phi^2) > 0.
\end{equation}
The exact solution of (\ref{first-order}) is the orbit $\{ \phi_0(x - s) \}_{s \in \mathbb{R}}$ with $\phi_0$ defined by 
\begin{equation}
    \label{exact-kink}
    \phi_0(x) = \left[ \frac{1}{2} (1 + \tanh(x)) \right]^{1/2}.
\end{equation}

If $\varepsilon \neq 0$ is a small parameter, we can get a family of solutions 
$\{ \phi_{\varepsilon}\}_{\varepsilon \in \mathbb{R}}$ to (\ref{ode}) by using the Lyapunov--Schmidt reduction method from the solution $\phi_0(x-s_0)$ with $\phi_0$ as in (\ref{exact-kink}) and with $s_0$ uniquely selected.  The following theorem gives the persistence result. 

\begin{theorem}
	\label{th-existence}
Assume that $V \in L^2(\mathbb{R})$ and that $s_0 \in \mathbb{R}$ is a simple root of the effective potential $V_{\rm eff}(s) : \R \to \R$ given by 
\begin{equation}
\label{eff-potential}
V_{\rm eff}(s) := \frac{1}{4} \int_{\mathbb{R}} V(x) {\rm sech}^2(x-s) dx.
\end{equation}
There exist $\varepsilon_0 > 0$ and $C_0 > 0$ such that for every 
$\varepsilon \in (-\varepsilon_0,\varepsilon_0)$, the stationary equation (\ref{ode}) admits the kink solution with the profile $\phi_{\varepsilon}$ satisfying
\begin{equation}
\label{bound-kink}
\| \phi_{\varepsilon} - \phi_0(\cdot - s_0) \|_{H^2} \leq C_0 |\varepsilon|,
\end{equation}
which implies that $\phi_{\varepsilon}(x) \to 0$ as $x \to -\infty$ and $\phi_{\varepsilon}(x) \to 1$ as $x \to +\infty$.
\end{theorem}

\begin{proof}
Let us represent the kink solution as $\phi_{\varepsilon}(x) = \phi_0(x-s) + w(x)$ and rewrite (\ref{ode}) as
\begin{equation}
\label{fixed-point}
L_+ w + \varepsilon V(x) (\phi_0(x-s) + w) + N(w) = 0, 
\end{equation}
where $L_+ : H^2(\R) \subset L^2(\R) \to L^2(\R)$ is given by the differential expression 
\begin{equation}
\label{operator-plus}
L_+ = -\partial_x^2 + 1 - 12 \phi_0^2(x-s) + 15 \phi_0^4(x-s)
\end{equation}
and $N(w)$ is at least quadratic near $0$ according to the expansion
$$
N(w) = w^2 (30 \phi_0^3(x-s) - 12 \phi_0(x-s)) + w^3 (30 \phi_0^2(x-s) - 4) + 15 w^4 \phi_0(x-s) + 3 w^5.
$$
For unique definition of $(w,s) \in H^2(\R) \times \mathbb{R}$, we add the constraint $\langle \phi_0'(\cdot - s), w \rangle_{L^2} = 0$ with the standard definition of the inner product in $L^2(\R)$.

Since $\sigma(L_+) = \{ 0 \} \cup [1,\infty)$ in $L^2(\R)$ with ${\rm Ker}(L_+) = {\rm span}(\phi'_0(\cdot - s))$ and $H^2(\R)$ is a Banach algebra with respect to multiplication, we can decompose the problem (\ref{fixed-point}) into two parts. The first part is a root-finding  equation for $F(w,s,\varepsilon) : H^2(\R) \times \R \times \R \to \R$:
\begin{equation}
\label{bif-eq}
F(w,s,\varepsilon) = 0, \quad 
F(w,s,\varepsilon) := \langle \phi'_0(\cdot - s), 
V (\phi_0(\cdot-s) + w) + \varepsilon^{-1} N(w) \rangle_{L^2}.
\end{equation}
The second part is given by the fixed-point equation for 
$G(w,s,\varepsilon) : H^2(\R) \times \R \times \R \to H^2(\R)$:
\begin{equation}
\label{non-eq}
w = G(w,s,\varepsilon), \quad 
G(w,s,\varepsilon) := -\Pi_0 L_+^{-1} \Pi_0 [\varepsilon V (\phi_0(\cdot - s) + w) + N(w)],
\end{equation}
where $\Pi_0 : L^2(\R) \to L^2(\R) |_{({\rm Ker}(L_+))^{\perp}}$ is the orthogonal projection operator. Since $\Pi_0 L_+^{-1} \Pi_0 : L^2(\R) \to H^2(\R)$ is bounded and $H^2(\R)$ is continuously embedded into $L^2(\R) \cap L^{\infty}(\R)$, there are $C_1,C_2,C_3 > 0$ such that for all $w \in H^2(\R)$, we have 
\begin{align*}
\| G(w,s,\varepsilon) \|_{H^2} &\leq C_1 \| \varepsilon V (\phi_0(\cdot - s) + w) + N(w) \|_{L^2} \\
&\leq C_2 |\varepsilon| 
\| V \|_{L^2} ( 1 + \| w \|_{H^2}) + C_3 \| w\|_{H^2}^2 (1 + \| w \|^3_{H^2}).
\end{align*}
By the Banach fixed-point theorem, there exists a unique fixed point 
$w \in H^2(\R)$ of (\ref{non-eq}) near $0$ for any $(s,\varepsilon)$ near $(s_0,0)$ for an arbitrary fixed $s_0 \in \R$. The unique solution $w \in H^2(\R)$  denoted as $w(x;s,\varepsilon)$ satisfies 
\begin{equation}
\label{bound-on-w}
\| w(\cdot;s,\varepsilon) \|_{H^2} \leq C_4 |\varepsilon|, \quad \varepsilon \in (-\varepsilon_0,\varepsilon_0),
\end{equation}
where $\varepsilon_0 > 0$ is sufficiently small and $C_4 > 0$ is independent of $s$ and $\varepsilon$ in $\mathbb{R} \times (-\varepsilon_0,\varepsilon_0)$. As a result, we get from (\ref{bif-eq}) and (\ref{bound-on-w}) that 
\begin{equation}
\label{bif-eq-exp}
F(w(\cdot;s,\varepsilon),s,\varepsilon) = V_{\rm eff}(s) + \tilde{F}(s,\varepsilon),
\end{equation}
where $V_{\rm eff}(s) : \R\to \R$ given by 
$$
V_{\rm eff}(s) = \langle \phi_0'(\cdot - s), 
V \phi_0(\cdot-s) \rangle_{L^2}
$$
is the effective potential and 
$\tilde{F}(s,\varepsilon) : \R \times \R\to \R$ is a smooth function of $(s,\varepsilon)$ satisfying the bound 
\begin{equation}
\label{F-tilde-bound}
|\tilde{F}(s,\varepsilon)| \leq C_5 |\varepsilon|, \quad  \varepsilon \in (-\varepsilon_0,\varepsilon_0),
\end{equation}
where $C_5 > 0$ is independent of $s$ and $\varepsilon$ in $\mathbb{R} \times (-\varepsilon_0,\varepsilon_0)$. Since (\ref{exact-kink}) implies that 
$$
\phi_0 \phi_0' = \frac{1}{2} \frac{d}{dx} \phi_0^2(x) = \frac{1}{4} {\rm sech}^2(x),
$$ 
the expression for $V_{\rm eff}(s)$ reduces to (\ref{eff-potential}). Let $s_0 \in \mathbb{R}$ be a simple root of $V_{\rm eff}$ such that 
\begin{equation}
\label{roots}
V_{\rm eff}(s_0) = 0, \quad V_{\rm eff}'(s_0) \neq 0.
\end{equation}
By using the expansion (\ref{bif-eq-exp}) with the bound (\ref{F-tilde-bound}), and the implicit function theorem, we obtain a unique solution $s \in \R$ to (\ref{bif-eq}) near $s_0$ for any $\varepsilon$ near $0$. The unique solution denoted as $s(\varepsilon)$ satisfies
\begin{equation}
\label{bound-on-s}
|s(\varepsilon) - s_0| \leq C_6 |\varepsilon|, \quad  \varepsilon \in (-\varepsilon_0,\varepsilon_0),
\end{equation}
where $C_6 > 0$ is independent of $\varepsilon$ in $(-\varepsilon_0,\varepsilon_0)$. We get (\ref{bound-kink}) from the bounds (\ref{bound-on-w}), (\ref{bound-on-s}), and the triangle inequality with
$$
\| \phi_0(\cdot - s(\varepsilon)) - \phi_0(\cdot - s_0) \|_{H^2} \leq C_7 
|s(\varepsilon) - s_0| \leq C_6 C_7 |\varepsilon|, \quad  \varepsilon \in (-\varepsilon_0,\varepsilon_0),
$$ 
where $C_7 > 0$ is independent of $\varepsilon$ in $(-\varepsilon_0,\varepsilon_0)$.
\end{proof}

\begin{remark}
	If $V(-x) = -V(x)$ is odd and smooth, then $s_0 = 0$ is a suitable solution to (\ref{roots}) provided that $\int_{\mathbb{R}} V'(x) {\rm sech}^2(x) dx \neq 0$. Nevertheless, we show in Theorem \ref{th-parity} that there is no odd parity of $v_{\varepsilon} := 2 \phi_{\varepsilon}^2 - 1$ for $\varepsilon \neq 0$ even if $V$ is odd. This suggests that $s(\varepsilon)$ in the bound (\ref{bound-on-s}) is generally nonzero for $\varepsilon \neq 0$ even if $s_0 = 0$ and $V$ is odd.
\end{remark}

\section{Kinks as minimizers or saddle points of energy}
\label{sec-3}

The NLS equation (\ref{NLS}) admits the conserved energy functional 
\begin{equation}
\label{energy}
E_{\varepsilon}(\psi) = \int_{\R} \left( |\psi_x|^2 + |\psi|^2 (1 - |\psi|^2)^2 + \varepsilon V |\psi|^2 \right) dx,
\end{equation}
which is well-defined in the energy space 
$$
\mathcal{E} = \left\{ \psi \in H^1_{\rm loc}(\mathbb{R}) : \quad \psi_x \in L^2(\R), \quad \psi (1-|\psi|^2) \in L^2(\R) \right\} \subset L^{\infty}(\R),
$$
provided that $V \in L^1(\R)$. 

The profile $\phi_{\varepsilon}$ of the kink solution in Theorem \ref{th-existence} satisfies the stationary equation (\ref{ode}), which is the Euler--Lagrange equation $E_{\varepsilon}'(\phi) = 0$ for the energy (\ref{energy}) in the space of real-valued functions. The second variation of $E_{\varepsilon}(\psi)$ at the kink profile $\phi_{\varepsilon}$ is defined in the space of real-valued functions by the self-adjoint operator $\mathcal{L}_+(\varepsilon) : H^2(\R) \subset L^2(\R) \to L^2(\R)$ given by the differential expression
\begin{equation}
\label{L-plus}
\mathcal{L}_+(\varepsilon) = -\partial_x^2 + 1 - 12 \phi_{\varepsilon}^2 + 15 \phi_{\varepsilon}^4 + \varepsilon V.
\end{equation}
The linear operator $\mathcal{L}_+(\varepsilon)$ is $\varepsilon$-close to the limiting operator $L_+$ given by  (\ref{operator-plus}) with $s = s(\varepsilon)$ being $\varepsilon$-close to $s_0$, see the bound (\ref{bound-on-s}). Since $\sigma(L_+) = \{ 0 \} \cup [1,\infty)$ in $L^2(\R)$ with ${\rm Ker}(L_+) = {\rm span}(\phi'_0(\cdot - s))$, 
one can anticipate that the $\varepsilon$-terms generate a perturbation which shifts the zero eigenvalue of $L_+$ to either a {\em positive} or a {\em negative} eigenvalue of $\mathcal{L}_+(\varepsilon)$ in $L^2(\R)$. In the former case, $\mathcal{L}_+(\varepsilon)$ is strictly positive and $\phi_{\varepsilon}$ is a strict minimizer of energy $E_{\varepsilon}(\psi)$ in the space of real-valued functions under the condition $V \in L^1(\R)$. In the latter case, $\mathcal{L}_+(\varepsilon)$ is sign-indefinite and $\phi_{\varepsilon}$ is a saddle point of energy $E_{\varepsilon}(\psi)$. The following theorem describes the splitting of the zero eigenvalue of $L_+$.

\begin{theorem}
	\label{th-splitting}
	Assume that $V \in H^1(\mathbb{R})$ and that $s_0 \in \mathbb{R}$ is a simple root of the effective potential $V_{\rm eff}$ given by (\ref{eff-potential}). Let $\{ \phi_{\varepsilon} \}_{\varepsilon \in (-\varepsilon_0,\varepsilon_0)}$ be the family of kink solutions satisfying (\ref{bound-kink}). Then, for every $\varepsilon \in (0,\varepsilon_0)$, the linear operator $\mathcal{L}_+(\varepsilon)$ is strictly positive if 
	$V_{\rm eff}'(s_0) < 0$ and admits a simple negative eigenvalue if 
	$V_{\rm eff}'(s_0) > 0$. The smallest eigenvalue of $\mathcal{L}_+(\varepsilon)$ denoted as $\mu_+(\varepsilon)$ satisfies $\mu_+(\varepsilon) = -4 \varepsilon V_{\rm eff}'(s_0) + \mathcal{O}(\varepsilon^2)$ as $\varepsilon \to 0$.
\end{theorem}

\begin{proof}
	To simplify computations of the perturbation theory, we adopt the decomposition $\phi_{\varepsilon}(x) = \phi_0(x-s) + w(x)$ from the proof of Theorem \ref{th-existence}, even though the parameter $s = s(\varepsilon)$ is uniquely defined in the neighborhood of $s_0$ and satisfies the bound (\ref{bound-on-s}). It follows from the fixed-point equation (\ref{non-eq}) that the solution 
	$w \in H^2(\R)$ can be represented by using the near-identity transformation in the form:
	$$
	w = \varepsilon w_1 + \tilde{w}_{\varepsilon}, 
	$$ 
	where $w_1 \in H^2(\R)$ is a unique solution of 
	\begin{equation}
\label{w-1}
	L_+ w_1 = -\Pi_0 V \phi_0(\cdot - s)
	\end{equation}
	satisfying $\langle \phi_0'(\cdot - s), w_1 \rangle_{L^2} = 0$ and $\tilde{w}_{\varepsilon} \in H^2(\R)$ is the remainder term satisfying the bound $\|\tilde{w}_{\varepsilon} \|_{H^2} \leq C \varepsilon^2$ for some $C > 0$ independently of $\varepsilon$ in $(-\varepsilon_0,\varepsilon_0)$. This allows us to expand
	\begin{equation}
\label{expansion-1}
	\mathcal{L}_+(\varepsilon) = L_+ + \varepsilon L_+^{(1)} +  \tilde{\mathcal{L}}_+(\varepsilon),
	\end{equation}
	where 
	$$
	L_+^{(1)} = V - 24 \phi_0(\cdot - s) w_1 + 60 \phi_0^3(\cdot - s) w_1,
	$$
	and $\tilde{\mathcal{L}}_+(\varepsilon)$ contains the remainder terms from $-12 \phi_{\varepsilon}^2 + 15 \phi_{\varepsilon}^4$ satisfying $\| \tilde{\mathcal{L}}_+(\varepsilon) \|_{H^2} \leq C \varepsilon^2$ for some $C > 0$ independently of $\varepsilon$ in $(-\varepsilon_0,\varepsilon_0)$. The linear equation (\ref{w-1}) can be rewritten explicitly as 
\begin{equation}
	\label{eq-w-1}
	(-\partial_x^2 + 1 - 12 \phi_0^2 + 15 \phi_0^4) w_1 = -V\phi_0 + \frac{\langle \phi_0',V \phi \rangle_{L^2}}{\| \phi_0'\|^2_{L^2}} \phi_0',
\end{equation}
	where $\phi_0 \equiv \phi_0(\cdot - s)$. Since $V \in H^1(\R)$, we have $w_1 \in H^3(\R)$ so that we are allowed to differentiate (\ref{eq-w-1}) in $x$ and to get
	$$
	(-\partial_x^2 + 1 - 12 \phi_0^2 + 15 \phi_0^4) w_1' + (60 \phi_0^3-24 \phi_0) \phi_0' w_1 = -V' \phi_0 - V \phi_0' + \frac{\langle \phi_0',V \phi \rangle_{L^2}}{\| \phi_0'\|^2_{L^2}} \phi_0''.
	$$
	Projecting this equation to $\phi_0' \equiv \phi_0'(\cdot - s)$ yields 
	\begin{equation}
	\label{projection-1}
	\langle L_+^{(1)} \phi_0', \phi_0' \rangle_{L^2} + \langle V' \phi_0, \phi_0' \rangle_{L^2} = 0,
	\end{equation}
	since $\langle \phi_0'', \phi_0' \rangle_{L^2}= 0$. The smallest isolated eigenvalue $\mu_+(\varepsilon)$  of the self-adjoint operator $\mathcal{L}_+(\varepsilon)$ in $L^2(\R)$ expanded as (\ref{expansion-1}) is obtained from the perturbation theory for small $\varepsilon \in (-\varepsilon_0,\varepsilon_0)$ as 
	\begin{align*}
	\mu_+(\varepsilon) &= \varepsilon \frac{\langle L_+^{(1)} \phi_0', \phi_0' \rangle_{L^2}}{\| \phi_0' \|^2_{L^2}} + \mathcal{O}(\varepsilon^2) \\
	&= -\varepsilon \frac{\langle V' \phi_0, \phi_0' \rangle_{L^2}}{\| \phi_0'\|^2_{L^2}} + \mathcal{O}(\varepsilon^2),
	\end{align*}
	where we have used (\ref{projection-1}). Since $H^1(\R)$ is embedded into the space of continuous and bounded functions, integration by parts yields 
	\begin{equation}
\label{projection-3}
	\langle V' \phi_0, \phi_0' \rangle_{L^2} = V_{\rm eff}'(s).
	\end{equation}
	On the other hand, we can compute explicitly
	$$
	\| \phi_0' \|^2_{L^2} = \frac{1}{8} \int_{\mathbb{R}} \frac{dx}{(1 + \tanh(x)) \cosh^4(x)} = \frac{1}{4}.
	$$
	By using the bound (\ref{bound-on-s}), we finally obtain 
\begin{align*}
	\mu_+(\varepsilon) = -4 \varepsilon V_{\rm eff}'(s_0) + \mathcal{O}(\varepsilon^2).
\end{align*}
Since $\sigma(L_+) = \{ 0 \} \cup [1,\infty)$, the smallest isolated eigenvalue $\mu_+(\varepsilon)$ is the only one located near $0$. For  small $\varepsilon \in (0,\varepsilon_0)$, $\mu_+(\varepsilon) > 0$ if $V'_{\rm eff}(s_0) < 0$ and $\mu_+(\varepsilon) < 0$ if $V'_{\rm eff}(s_0) > 0$, which yields the assertion.
\end{proof}

\begin{remark}
	\label{rem-minimizer}
The constraint $V_{\rm eff}'(s_0) \neq 0$ in Theorem \ref{th-splitting} is the same as in Theorem \ref{th-existence}. However, we have to use a more restrictive assumption $V \in H^1(\R)$ in Theorem \ref{th-splitting} compared to $V \in L^2(\R)$ in Theorem \ref{th-existence} to justify differentiating (\ref{eq-w-1}) in $x$. In addition, if we want to place the conclusion 
of Theorem \ref{th-splitting} in the variational context, then  we need to restrict $V$ under the additional condition $V \in L^1(\R)$ so that the energy functional $E_{\varepsilon}(\psi)$ is well-defined in the energy space $\mathcal{E}$. Under this restriction, $V \in H^1(\R) \cap L^1(\R)$, 
the kink profile $\phi_{\varepsilon}$ is a local energy minimizer of $E_{\varepsilon}(\psi)$ if $V_{\rm eff}'(s_0) < 0$ and a saddle point of $E_{\varepsilon}(\psi)$ if $V_{\rm eff}'(s_0) > 0$.
\end{remark}

\section{Spectral stability of the pinned kinks}
\label{sec-4}

The spectral stability of the pinned kinks of Theorem \ref{th-existence} is determined by eigenvalues of the spectral stability problem 
\begin{equation}
\label{spec-stab}
\mathcal{L}_+(\varepsilon) u = -\lambda v, \quad 
\mathcal{L}_-(\varepsilon) v = +\lambda u,
\end{equation}
where $(u+iv) \in H^2(\mathbb{R})$ is the complex-valued perturbation to the kink profile $\phi_{\varepsilon}$ of Theorem \ref{th-existence}, the self-adjoint operator $\mathcal{L}_+(\varepsilon) : H^2(\R) \subset L^2(\R) \to L^2(\R)$ is defined by (\ref{L-plus}), and the self-adjoint operator $\mathcal{L}_-(\varepsilon) : H^2(\R) \subset L^2(\R) \to L^2(\R)$ is defined by  
\begin{equation}
\label{L-minus}
\mathcal{L}_-(\varepsilon) = -\partial_x^2 + 1 - 4 \phi_{\varepsilon}^2 + 3 \phi_{\varepsilon}^4 + \varepsilon V.
\end{equation}
The linear operator $\mathcal{L}_-(\varepsilon)$ is $\varepsilon$-close to the limiting operator $L_- = -\partial_x^2 + 1 - 4 \phi_0^2 + 3 \phi_0^4$ for which we have $\sigma(L_-) = [0,\infty)$ in $L^2(\R)$. Since $\mathcal{L}_-(\varepsilon) \phi_{\varepsilon} = 0$ and 
$\phi_{\varepsilon} \notin L^2(\R)$, $0$ is the end point of the essential spectrum of $\mathcal{L}_-(\varepsilon)$. This creates a challenge for the perturbation theory of small eigenvalues of the spectral problem (\ref{spec-stab}) bifurcating from the zero eigenvalue. To deal with the challenge, we need the following technical result, which is proven under the additional requirement $V \in L^1(\R)$. The constraint $V \in L^1(\R)$ is consistent with the restriction of $V$ in  Remark \ref{rem-minimizer} for the energy functional (\ref{energy}) to be defined in the energy space $\mathcal{E}$.

\begin{lemma}
	\label{lem-Evans}
	Let $\lambda \in \mathbb{C}$ be an isolated eigenvalue of the spectral problem (\ref{spec-stab}) with ${\rm Re}(\lambda) > 0$ 
	and $|\lambda| \leq 1$. If $V \in L^2(\R) \cap L^1(\R)$, then the eigenfunction $(u,v) \in H^2(\R) \times H^2(\R)$ satisfies the asymptotic behavior 
	\begin{equation}
	\label{minus-inf}
	\left( \begin{matrix} u \\ v \end{matrix} \right) \in {\rm span}\left( 	\left( \begin{matrix} -\lambda \\ 1-\kappa_+^2 \end{matrix} \right) e^{\kappa_+ x}, \left( \begin{matrix} -\lambda  \\ 1 - \kappa_-^2 \end{matrix} \right)  e^{\kappa_- x} \right) \quad \mbox{\rm as} \;\; x \to -\infty
	\end{equation}
	with $\kappa_{\pm} = \sqrt{1 \pm i \lambda} = 1 + \mathcal{O}(\lambda)$ as $\lambda \to 0$, and 
	\begin{equation}
	\label{plus-inf}
	\left( \begin{matrix} u \\ v \end{matrix} \right) \in {\rm span}\left( 	\left( \begin{matrix} -\lambda \\ 4-\nu_+^2 \end{matrix} \right) e^{-\nu_+ x}, \left( \begin{matrix} -\lambda  \\ 4 - \nu_-^2 \end{matrix} \right)  e^{-\nu_- x} \right) \quad \mbox{\rm as} \;\; x \to +\infty
	\end{equation}
	with $\nu_+ = \sqrt{2 + \sqrt{4 - \lambda^2}} = 2 + \mathcal{O}(\lambda^2)$ and $\nu_- = \sqrt{2 - \sqrt{4 - \lambda^2}} = \frac{1}{2} \lambda + \mathcal{O}(\lambda^3)$ as $\lambda \to 0$.	
\end{lemma}

\begin{proof}
We can rewrite the spectral problem (\ref{spec-stab}) as the system of differential equations: 
\begin{equation}
\label{linear-de}
	\begin{cases}
	-u'' + u - 12 \phi_{\varepsilon}^2 u + 15 \phi_{\varepsilon}^4 u + \varepsilon V u = -\lambda v, \\
	-v'' + v - 4 \phi_{\varepsilon}^2 v + 3 \phi_{\varepsilon}^4 v + \varepsilon V v = +\lambda u.
	\end{cases}
\end{equation}
	If $V \in L^2(\R)$, then $\phi_{\varepsilon}$ exists by Theorem \ref{th-existence} with the boundary conditions $\phi_{\varepsilon}(x) \to 0$ as $x \to -\infty$ and $\phi_{\varepsilon}(x) \to 1$ as $x \to +\infty$. By Coddington--Levinson theory (see \cite[Proposition 8.1 on p.92]{CL55}), if $V \in L^1(\R)$, then the convergence of $\phi_{\varepsilon}$ as $x \to \pm \infty$ is exponential and solutions of the system  (\ref{linear-de}) for $\lambda \in \mathbb{C}$ with ${\rm Re}(\lambda) > 0$ converge as $x \to \pm \infty$  to the exponential solutions of the linear system of differential equations with 
	constant coefficients. As $x \to -\infty$, where $\phi_{\varepsilon}(x) \to 0$,  the characteristic equation of this linear system is 
	$$
	0 = \left| \begin{matrix} 1 - \kappa^2 & \lambda \\ -\lambda & 1 - \kappa^2 \end{matrix} \right| = (1-\kappa^2)^2 + \lambda^2,
	$$
	with four roots $\pm \kappa_+$ and $\pm \kappa_-$ in the exponentially decaying functions (\ref{minus-inf}). As $x \to +\infty$, where $\phi_{\varepsilon}(x) \to 1$,  the characteristic equation of this linear system is 
	$$
	0 = \left| \begin{matrix} 4 - \nu^2 & \lambda \\ -\lambda & - \nu^2 \end{matrix} \right| = \nu^4 - 4 \nu^2 + \lambda^2,
	$$
	with four roots $\pm \nu_+$ and $\pm \nu_-$ in the exponentially decaying functions (\ref{plus-inf}). Expanding the roots as $\lambda \to 0$ yields the Taylor expansions for $\kappa_{\pm}$ and $\nu_{\pm}$, where the choice of the square root for $\nu_-$ provides the exponential decay of eigenfunctions in (\ref{plus-inf}) with ${\rm Re}(\nu_-) > 0$ for ${\rm Re}(\lambda) > 0$.
\end{proof}

When the {\em zero} eigenvalue of $L_+$ in $L^2(\R)$ becomes a {\em negative} eigenvalue of $\mathcal{L}_+(\varepsilon)$ in $L^2(\R)$ for  $\varepsilon \neq 0$ by Theorem \ref{th-splitting} in the case of $V_{\rm eff}'(s_0) > 0$, we can anticipate the linear instability of the pinned kink and the existence of an isolated eigenvalue $\lambda = \lambda(\varepsilon) \in \mathbb{C}$ with ${\rm Re}(\lambda(\varepsilon)) > 0$ in the spectral stability problem (\ref{spec-stab}) for small $\varepsilon > 0$. The following theorem describes the spectral instability of the corresponding pinned kink and clarifies the leading-order asymptotics of the unstable eigenvalue $\lambda(\varepsilon)$ as $\varepsilon \to 0$.

\begin{theorem}
	\label{th-stability}
	Assume that $V \in H^1(\mathbb{R}) \cap L^1(\R)$ and that $s_0 \in \mathbb{R}$ is a simple root of the effective potential $V_{\rm eff}$ given by (\ref{eff-potential}) such that 	$V_{\rm eff}'(s_0) > 0$. Let $\{ \phi_{\varepsilon} \}_{\varepsilon \in (-\varepsilon_0,\varepsilon_0)}$ be the family of kink solutions satisfying (\ref{bound-kink}). Then, the spectral stability problem 
	(\ref{spec-stab}) with $\varepsilon \in (0,\varepsilon_0)$ admits a real 
	unstable eigenvalue $\lambda = \lambda(\varepsilon)$ such that $\lambda(\varepsilon) = 2 \varepsilon V_{\rm eff}'(s_0) + \mathcal{O}(\varepsilon^2)$ as $\varepsilon \to 0$.
\end{theorem}

\begin{proof}
We use the expansion (\ref{expansion-1}) for $\mathcal{L}_+(\varepsilon)$ defined by (\ref{L-plus}). Similarly for $\mathcal{L}_-(\varepsilon)$ defined by (\ref{L-minus}), we expand 
\begin{equation}
\label{expansion-2}
\mathcal{L}_-(\varepsilon) = L_- + \varepsilon L_-^{(1)} +  \tilde{\mathcal{L}}_-(\varepsilon),
\end{equation}
where 
$$
L_-^{(1)} = V - 8 \phi_0(\cdot - s) w_1 + 12 \phi_0^3(\cdot - s) w_1,
$$
and $\tilde{\mathcal{L}}_-(\varepsilon)$ contains the remainder terms from $-4 \phi_{\varepsilon}^2 + 3 \phi_{\varepsilon}^4$ satisfying $\| \tilde{\mathcal{L}}_-(\varepsilon) \|_{H^2} \leq C \varepsilon^2$ for some $C > 0$ independently of $\varepsilon$ in $(-\varepsilon_0,\varepsilon_0)$. 

We obtain a solution $(u,v) \in H^2(\R) \times H^2(\R)$ of the spectral stability problem (\ref{spec-stab}) for an eigenvalue $\lambda = \lambda(\varepsilon)$ satisfying ${\rm Re}(\lambda) > 0$ by using the following asymptotic expansion:
\begin{equation}
\label{expansion-eigenvector}
\begin{cases}
u = \phi_0' + \varepsilon u_1 + \mathcal{O}(\varepsilon^2), \\
v = c \phi_0 + \varepsilon v_1 + \mathcal{O}(\varepsilon^2), \\
\lambda = \varepsilon \lambda_1 + \mathcal{O}(\varepsilon^2),
\end{cases}
\end{equation}
where $\phi_0 \equiv \phi_0(\cdot - s)$,  $\phi_0' \equiv \phi_0'(\cdot - s)$, the coefficient $c \in \mathbb{C}$ is to be defined, and the first-order correction $(u_1,v_1,\lambda_1) \in H^2(\R) \times H^2(\R) \times \C$ is  to be found. Since $L_+ \phi'_0 = 0$ and $L_- \phi_0 = 0$, we obtain from the expansions (\ref{expansion-1}), (\ref{expansion-2}), and (\ref{expansion-eigenvector}) in the spectral problem (\ref{spec-stab})  
that the leading-order equations at the order of $\mathcal{O}(\varepsilon)$ are given by 
\begin{equation}
\label{system-1}
\mathcal{O}(\varepsilon) : \quad \begin{cases}
L_+ u_1 + L_+^{(1)}  \phi_0' = - c \lambda_1 \phi_0, \\
L_- v_1 + c L_-^{(1)}  \phi_0 = \lambda_1 \phi_0'. 
\end{cases}
\end{equation}

Since expansion of $\mathcal{L}_-(\varepsilon) \phi_{\varepsilon} = 0$ at the order of $\mathcal{O}(\varepsilon)$ yields 
$$
L_- w_1 + L_-^{(1)}  \phi_0 = 0,
$$
where $w_1$ is a solution of (\ref{w-1}), we obtain a closed solution of the second equation of the system (\ref{system-1}) in the form
$$
v_1 = c w_1 - \frac{1}{2} \lambda_1 (x-s) \phi_0 + c_1 \phi_0 + c_2 v_0,
$$
where $\{ \phi_0,v_0\}$ are two linearly independent solutions 
of the second-order equation $L_- v = 0$ and $\{ c_1,c_2\}$ are arbitrary constants. Since solutions of $L_- v = 0$ are spanned by $\{ e^x, e^{-x}\}$ as $x \to -\infty$ and by $\{ 1,x\}$ as $x \to +\infty$, it follows from the fact that $\phi_0(x) \to 0$ converges exponentially as $x \to -\infty$ that 
$v_0(x)$ diverges exponentially as $x \to -\infty$. Hence, we set $c_2 = 0$ to get the exponential convergence of $v_1(x) \to 0$ as $x \to -\infty$.

The expansion for $v$ is written so far as 
$$
v = \phi_0 \left( c -\frac{1}{2} \varepsilon \lambda_1 (x-s) + \varepsilon c_1 \right) + c \varepsilon w_1 + \mathcal{O}(\varepsilon^2),
$$
where $w_1 \in H^2(\R)$ decays to $0$ at both infinities. Since $\phi_0(x) \to 1$ as $x \to +\infty$, we obtain the expansions for large $x$:
\begin{equation}
\label{limit-1}
v(x) \sim c -\frac{1}{2} \varepsilon \lambda_1 (x-s) + \varepsilon c_1 + \mathcal{O}(\varepsilon^2) \quad \mbox{\rm as} \;\; x \to +\infty.
\end{equation}
By Lemma \ref{lem-Evans}, we have $\nu_- = \frac{1}{2} \lambda + \mathcal{O}(\lambda^3) \to 0$ as $\lambda \to 0$ if ${\rm Re}(\lambda) > 0$ (to be determined). The spanning set (\ref{plus-inf}) suggests that the $v$ component is expanded as 
\begin{equation}
\label{limit-2}
v(x) \sim c e^{-\nu_- x} = c \left( 1 - \frac{1}{2} \lambda x + \mathcal{O}(\lambda^2 x^2) \right)
\end{equation}
since $e^{-\nu_+ x}$ with $\nu_+ = 2 + \mathcal{O}(\lambda^2)$ converges to $0$ exponentially fast as $x \to +\infty$. Comparing (\ref{limit-1}) and (\ref{limit-2}), we conclude that $c = 1$ and $c_1 = -\frac{1}{2} s_0$ 
since $s(\varepsilon) = s_0 + \mathcal{O}(\varepsilon)$ from the bound (\ref{bound-on-s}) and $\lambda = \varepsilon \lambda_1 + \mathcal{O}(\varepsilon^2)$.

Projecting the first equation of the system (\ref{system-1}) to $\phi_0'$ yields 
\begin{equation}
\label{projection-2}
\langle L_+^{(1)} \phi_0', \phi_0' \rangle_{L^2} = -c \lambda_1 \langle \phi_0, \phi_0' \rangle.
\end{equation}
Since $c = 1$ and $\langle \phi_0, \phi_0' \rangle_{L^2}= \frac{1}{2}$, we use (\ref{projection-1}), (\ref{projection-3}), and (\ref{projection-2}) to obtain 
$$
\lambda_1 = -2 \langle L_+^{(1)} \phi_0', \phi_0' \rangle_{L^2} 
= 2 \langle V' \phi_0,\phi_0' \rangle = 2 V_{\rm eff}'(s).
$$
Since $V_{\rm eff}'(s) = V_{\rm eff}'(s_0) + \mathcal{O}(\varepsilon)$ and $V_{\rm eff}'(s_0) > 0$, we confirm that ${\rm Re}(\lambda) > 0$ 
for the eigenvalue $\lambda = \varepsilon \lambda_1 + \mathcal{O}(\varepsilon^2)$. Persistence of this isolated eigenvalue outside the continuous spectrum of the spectral problem (\ref{spec-stab}) for $\varepsilon \in (0,\varepsilon_0)$ is proven by the standard methods of the perturbation theory. 
\end{proof}

\begin{remark}
	If $\varepsilon \neq 0$, the translational symmetry is broken and this leads to the bifurcation of a nonzero eigenvalue $\lambda(\varepsilon)$ in the spectral stability problem (\ref{spec-stab}) from the zero eigenvalue that exists for $\varepsilon = 0$. If the zero eigenvalue at $\varepsilon = 0$ is isolated from the continuous spectrum, as happens for bright solitons \cite{Kapitula}, then the standard Puiseux expansions provide approximations of the perturbed eigenvalue $\lambda(\varepsilon)$ with the leading-order behavior $\lambda(\varepsilon) = \mathcal{O}(\varepsilon^{1/2})$. However, the zero eigenvalue is embedded into the continuous spectrum for the kinks and we have shown in Theorem \ref{th-stability} that the leading-order behavior of the unstable eigenvalue is given by $\lambda(\varepsilon) = \mathcal{O}(\varepsilon)$. 
\end{remark}

\begin{remark}
	By the Hamiltonian symmetry of isolated eigenvalues of the spectral stability problem (\ref{spec-stab}), if there exists the unstable eigenvalue $\lambda(\varepsilon) > 0$ with the eigenfunction $(u,v) \in H^2(\R) \times H^2(\R)$, then there exists another stable eigenvalue $-\lambda(\varepsilon) < 0$ with the eigenfunction $(u,-v) \in H^2(\R) \times H^2(\R)$ so that a pair of real eigenvalues $\pm \lambda(\varepsilon)$ exists in the case of $V_{\rm eff}'(s_0) > 0$.
\end{remark}

\begin{remark}
	\label{rem-stab}
	If $V_{\rm eff}'(s_0) < 0$, the leading-order asymptotics of $\lambda(\varepsilon) = 2 \varepsilon V_{\rm eff}'(s_0) + \mathcal{O}(\varepsilon^2)$ as $\varepsilon \to 0$ would imply $\lambda(\varepsilon) < 0$. However, the derivation of the leading-order asymptotics relies on Lemma \ref{lem-Evans} which only holds for ${\rm Re}(\lambda) > 0$. As a result, we cannot justify the asymptotic expansion (\ref{expansion-eigenvector}) and prove the existence of a real negative eigenvalue in the case of  $V_{\rm eff}'(s_0) < 0$.
\end{remark}

\begin{remark}
	\label{rem-particle}
	Extrapolating the linear instability result of Theorem \ref{th-stability}, we conjecture that the nonlinear dynamics of the modulated kink solutions $\{ \phi_{\varepsilon}(\cdot - s(t))\}_{t > 0}$ is given by the viscous particle equation 
\begin{equation}
\label{viscous-particle}
	\frac{ds}{dt} = 2 \varepsilon V_{\rm eff}(s), \quad t > 0.
\end{equation}
The particle equation (\ref{viscous-particle}) describes the source $s_0$ if $V_{\rm eff}(s_0) = 0$, $V_{\rm eff}'(s_0) > 0$ and the sink $s_0$ if $V_{\rm eff}(s_0) = 0$, $V_{\rm eff}'(s_0) < 0$ for $\varepsilon \in (0,\varepsilon_0)$. The problem of asymptotic stability of the kink in the case of  $V_{\rm eff}'(s_0) < 0$ remains open for further studies.
\end{remark}

\section{Lack of parity symmetry for kinks}
\label{sec-5}

The exact solution (\ref{exact-kink}) for $\varepsilon = 0$ admits the odd symmetry of $\phi_0^2$ relative to the level $\frac{1}{2}$ as is seen from the explicit expression
$$
v_0(x) := 2\phi_0^2(x) - 1 = \tanh(x) = -\tanh(-x).
$$
Here, we show that no odd parity symmetry of $v_{\varepsilon} := 2 \phi_{\varepsilon}^2 - 1$ is preserved for $\varepsilon \neq 0$ even if $V$ is odd. We start with a technical lemma.

\begin{lemma}
	\label{lem-constraint}
Assume that $V \in L^2(\mathbb{R})$ and that there exists a family $\{ \phi_{\varepsilon} \}_{\varepsilon \in (-\varepsilon_0,\varepsilon_0)}$ of kink solutions in Theorem \ref{th-existence}. Then, 
\begin{equation}
\label{existence}
\int_{\R} V(x) \phi_{\varepsilon}(x) \phi_{\varepsilon}'(x) dx = 0.
\end{equation}
\end{lemma}

\begin{proof}
It follows from (\ref{ode}) that 
$$
\frac{d}{dx} \left[ (\phi')^2 - \phi^2 (1-\phi^2)^2 \right] = 2 \varepsilon V \phi \phi'. 
$$
Due to the boundary conditions $\phi(x) \to 0$ as $x \to -\infty$, we can write 
\begin{equation}
\label{first}
(\phi')^2 - \phi^2 (1-\phi^2)^2 = 2 \varepsilon \int_{-\infty}^x V(y) \phi(y) \phi'(y) dy,
\end{equation}
where the integral is well defined in $L^{\infty}(\R)$ since $V, \phi \phi' \in L^2(\R)$. Using the boundary condition $\phi(x) \to 1$ as $x \to +\infty$ yields the constraint (\ref{existence}). 
\end{proof}

\begin{remark}
	If $V(x) \geq 0$ and $\phi'(x) \geq 0$ for all $x \in \R$, then the existence condition (\ref{existence}) is not satisfied and no kinks exist.
\end{remark}

\begin{remark}
	If $V$ is odd and $v := 2\phi^2 - 1$ is odd, then $v' = 4 \phi \phi'$ is even and the existence condition (\ref{existence}) is satisfied. However, we show in Theorem \ref{th-parity} that no odd solution exists for $v$ even if $V$ is odd.
\end{remark}

\begin{theorem}
	\label{th-parity} 
Assume that $V \in H^1(\mathbb{R})$ is odd and that there exists a family $\{ \phi_{\varepsilon} \}_{\varepsilon \in (-\varepsilon_0,\varepsilon_0)}$ of kink solutions in Theorem \ref{th-existence}. Then, $v_{\varepsilon} := 2 \phi_{\varepsilon}^2 - 1$ is not odd.
\end{theorem}

\begin{proof}
Substituting $\psi := \phi^2$ into (\ref{ode}) and using (\ref{first}) yields 
\begin{align*}
\psi'' &= 2 (\phi')^2 + 2\psi (1- \psi) (1- 3 \psi) + 2 \varepsilon V \psi \\
&= 2 \psi (1-\psi)^2 + 2\psi (1- \psi) (1- 3 \psi) + 2 \varepsilon V \psi + 2 \varepsilon \int_{-\infty}^x V(y) \psi'(y) dy\\
&= 4 \psi (1-\psi) (1- 2 \psi) + 2 \varepsilon V \psi + 2 \varepsilon \int_{-\infty}^x V(y) \psi'(y) dy.
\end{align*}
Substitute $v := 2\psi - 1$ yields a closed equation,
\begin{align}
\label{closed}
v'' = -2 (1-v^2) v + 2 \varepsilon V (1+v) + 2 \varepsilon \int_{-\infty}^x V(y) v'(y) dy.
\end{align}
For $\varepsilon = 0$, this equation $v'' = -2(1-v^2) v$ 
yields the odd solution $v_0(x) = \tanh(x)$. If $\varepsilon \neq 0$, we assume that both $V$ and $v$ are odd 
and obtain a contradiction from (\ref{closed}). Indeed, if both $V$ and $v$ are odd, then $-2 (1-v^2) v$ and $2 \varepsilon V$ are odd, whereas $2 \varepsilon V v$ and $2\varepsilon \int_{-\infty}^x V(y) v'(y) dy$ are even. To show that the integral term is even, we denote $I(x) := \int_{-\infty}^x V(y) v'(y) dy$ and obtain 
\begin{align*}
I(-x) &= \int_{-\infty}^{-x} V(y) v'(y) dy = -\int_{+\infty}^{x} V(-y) v'(-y) dy \\
&= \int_{+\infty}^x V(y) v'(y) dy = \int_{-\infty}^x V(y) v'(y) dy = I(x),
\end{align*}
where we have used the constraint (\ref{existence}) and the evenness of $v'$. Furthermore, the even terms in (\ref{closed}) are nonzero, since if they are zero, we would get 
\begin{align}
\label{integral-const}
V v + \int_{-\infty}^x V(y) v'(y) dy = 0.
\end{align}
If $V \in H^1(\R)$, then $V \in C^0(\R) \cap L^{\infty}(\R)$ and $W \in C^1(\R)$, where $W(x) := \int_{-\infty}^x V(y) v'(y) dy$. Therefore, we get 
\begin{align*}
\frac{v}{v'} W' + W = 0 \quad 
\Rightarrow\quad W = \frac{c}{v}  \quad 
\Rightarrow \quad V = -\frac{c}{v^2},
\end{align*}
where $c$ is an arbitrary constant. However, the only constant for the odd function $V v^2$ is $0$, hence $V \equiv 0$ or $v \equiv 0$ are required to satisfy (\ref{integral-const}). The former case is trivial and the latter case contradicts the convergence $v(x) \to \pm 1$ as $x \to \pm \infty$. Thus, the right-hand side of (\ref{closed}) is not closed in the space of odd functions if $V$ and $v$ are odd, since the even term 
$2 \varepsilon (Vv + I)$ is nonzero. Hence, $v$ is not odd.
\end{proof}

\section{Numerical computations}
\label{sec:numerics}

We now present numerical computations that illustrate the persistence and stability of kinks in external potentials. For the first two examples we use
\begin{align}
\label{potential-1}
V(x)=\pm x e^{-x^2}, \quad \varepsilon = 1,
\end{align}
for which the value of $\varepsilon$ is not small.  For the remaining examples we use
\begin{align}
\label{potential-2}
V(x) = (1-a x^2){\rm sech}^2(x), \quad \varepsilon = 0.1.
\end{align}

Before we delve into the numerical details, it is also worth
highlighting the following point.
The stationary kink is computed on a large finite interval
with the second derivative approximated via
a centered finite-difference discretization and using a Newton iteration
scheme. The boundary values $\phi=0$ and $\phi=1$ are set as the Dirichlet conditions at the left and right endpoints, respectively.  
The linear operators $\mathcal{L}_+(\varepsilon)$ and $\mathcal{L}_-(\varepsilon)$ are then constructed using the numerical approximation for the kink profile.  For the spectral stability problem~(\ref{spec-stab}), one needs to take special care of controlling the computational domain due to the 
slow exponential decay of eigenfunctions as $x\to+\infty$,
\[
\nu_-=\frac{\lambda}{2}+\mathcal{O}(\lambda^3),
\]
prescribed by Lemma~\ref{lem-Evans}.  Hence, when the unstable eigenvalue $\lambda$ is small (as will be seen to be the case in our
numerical examples below), its eigenfunction has a large characteristic decay length of order $2 \lambda^{-1}$. Given the divergence of this length scale as $\lambda \rightarrow 0$, we opt to compute the isolated eigenvalues on an extended nonuniform domain whose right endpoint is placed several slow-tail decay lengths away from the kink.  This avoids the displacement of a small real eigenvalue caused by a fixed computational box.

\begin{figure}[htbp]
	\centering
	\includegraphics[width=0.475\linewidth,height=0.4\textwidth]{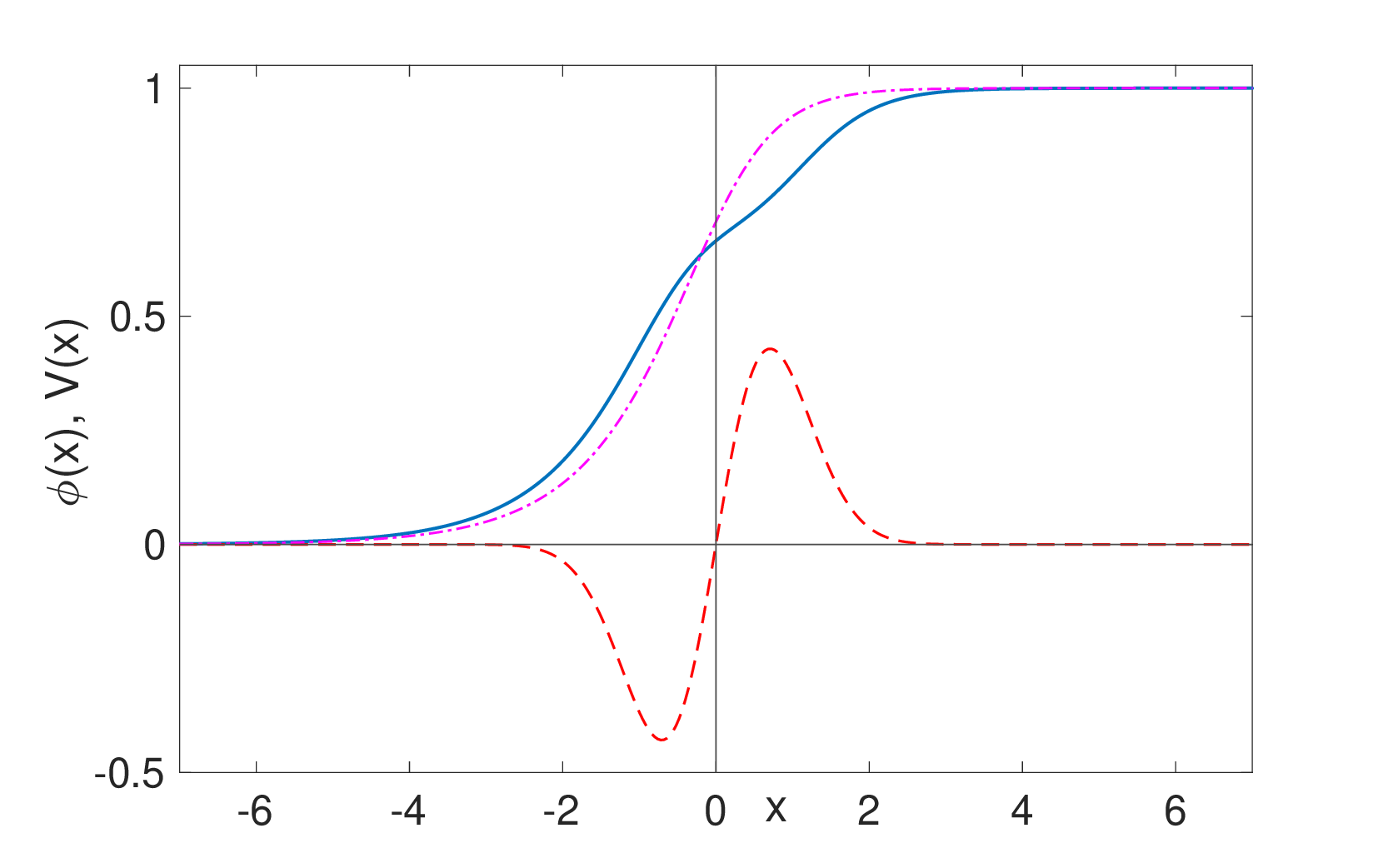}
	\includegraphics[width=0.475\linewidth,height=0.4\textwidth]{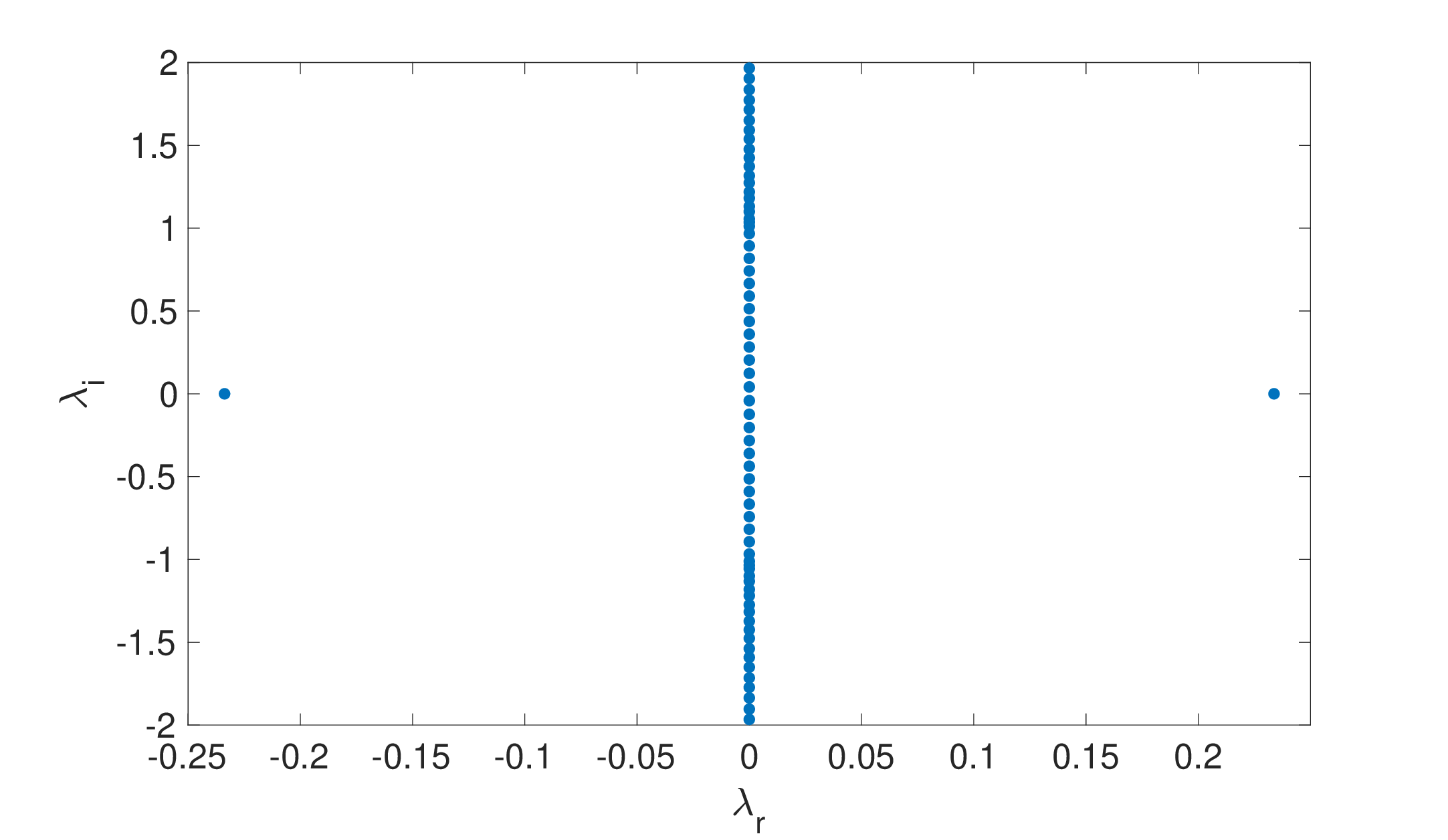}
	\caption{Numerical results for the potential $V(x)=xe^{-x^2}$. Left:  the perturbed kink profile $\phi$ (thick solid line), the potential $V$ (red dashed line), and the unperturbed kink $\phi_0$ for comparison.  Right: the complex plane $(\lambda_r,\lambda_i)$ of eigenvalues
		$\lambda=\lambda_r + i \lambda_i$ of the spectral stability problem (\ref{spec-stab}).}
	\label{fig:odd-plus}
\end{figure}

\begin{figure}[htbp]
	\centering
	\includegraphics[width=0.475\linewidth,height=0.4\textwidth]{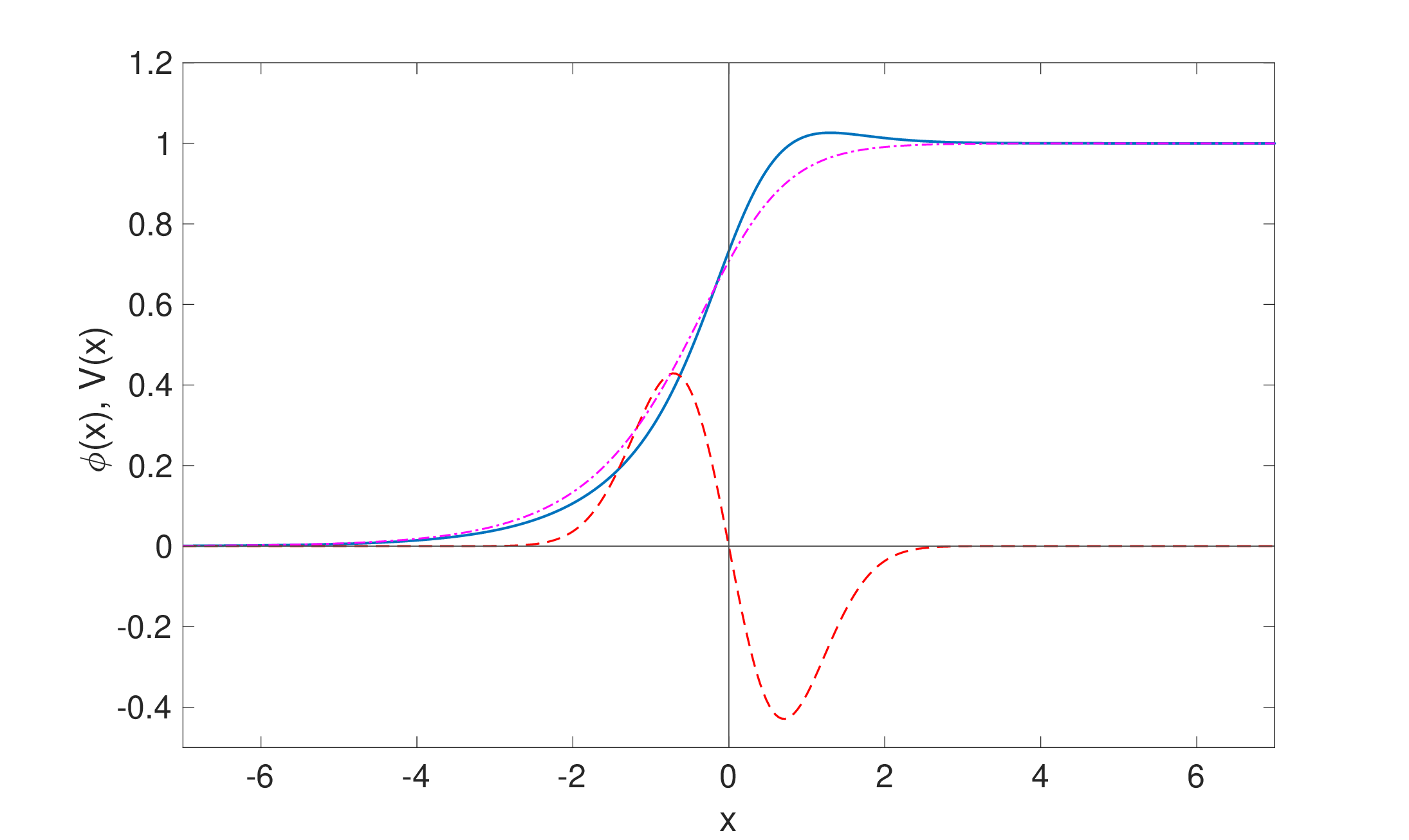}
	\includegraphics[width=0.475\linewidth,height=0.4\textwidth]{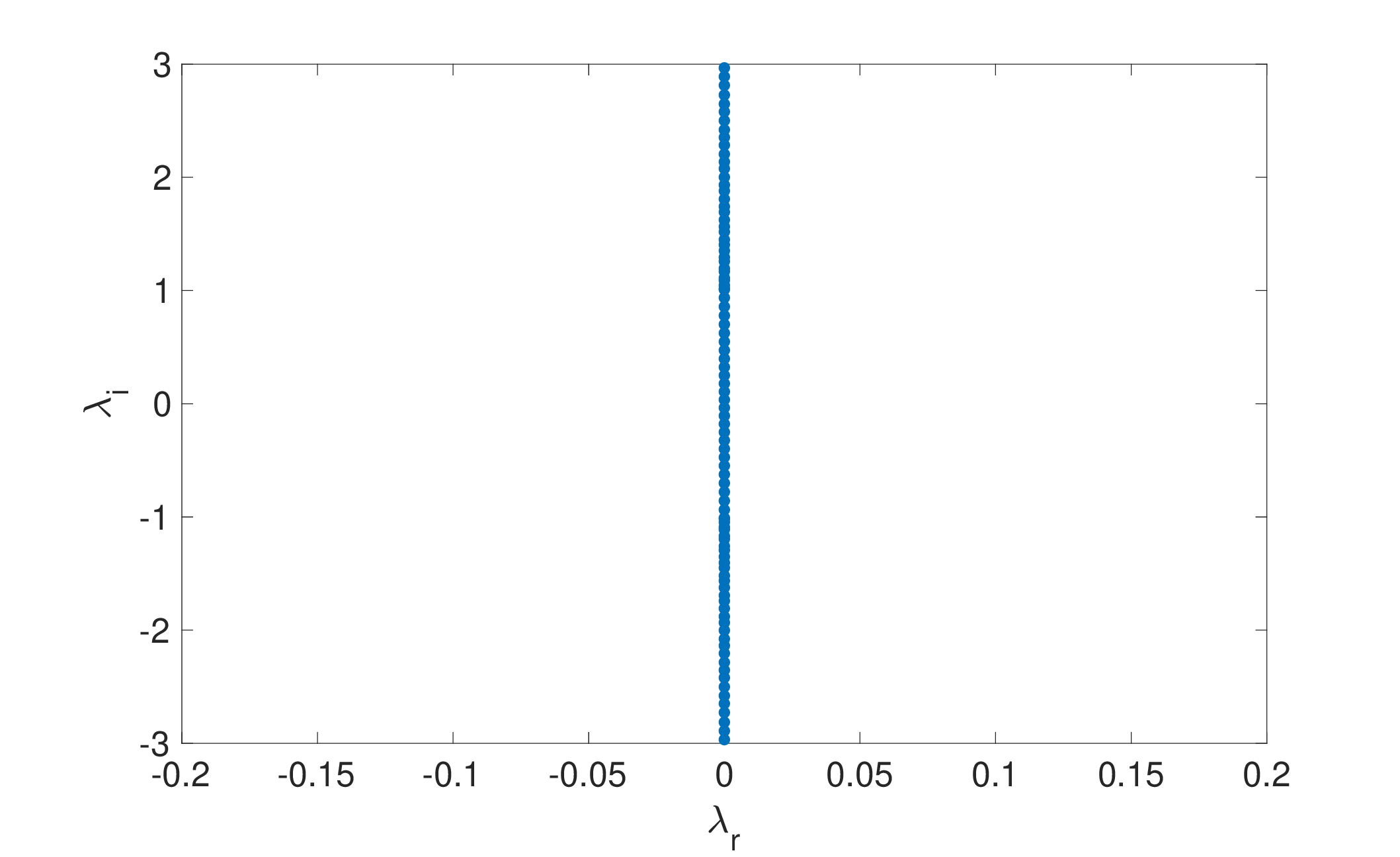}
	\caption{The same as Figure \ref{fig:odd-plus} but for the potential $V(x) = -xe^{-x^2}$. }
	\label{fig:odd-minus}
\end{figure}

\subsection{Kinks in the potential (\ref{potential-1})}

Figure~\ref{fig:odd-plus} presents the result for $V(x)=x e^{-x^2}$.  Since the potential $V$ is odd, $s_0=0$ is the root of $V_{\rm eff}$ in Theorem \ref{th-existence}.  Direct evaluation gives $V_{\rm eff}'(0) > 0$ so that 
Theorem~\ref{th-splitting} predicts that the translational eigenvalue (formerly at $0$ in the unperturbed case $\varepsilon = 0$)
becomes a negative eigenvalue of $\mathcal{L}_+(\varepsilon)$ for $\varepsilon > 0$ given approximately by 
\begin{equation}
\label{eig-L-plus}
\mu_+(\varepsilon) \approx - 4 \varepsilon V_{\rm eff}'(s_0).
\end{equation}
Theorem~\ref{th-stability} predicts a real unstable eigenvalue of the spectral stability problem (\ref{spec-stab}) given  approximately by 
\begin{equation}
\label{eig-spec-stab}
\lambda(\varepsilon) \approx + 2 \varepsilon V_{\rm eff}'(s_0).
\end{equation}
Theorem \ref{th-parity} suggests that the kink profile $\phi$ does not satisfy the odd parity symmetry in the component $v = 2 \phi^2 - 1$.

All of these analytical results are in direct parallel to numerical observations. The relevant numerical eigenvalue of $\mathcal{L}_+$ is approximately $\mu_+ \approx -0.613$, while the positive numerical eigenvalue within the real eigenvalue pair is $\lambda\approx0.234$.
If the leading-order formulas are evaluated from (\ref{eig-L-plus}) and (\ref{eig-spec-stab}) with $\varepsilon = 1$, they give $\mu_+\approx-0.575$ and $\lambda\approx0.288$.
While the asymptotic formulas are not expected to be quantitatively sharp here since $\varepsilon = 1$, they correctly predict both the sign and the order of magnitude for the relevant eigenvalues. In addition, it is clear from Figure \ref{fig:odd-plus} that the component $v = 2 \phi^2 - 1$ is not odd in $x$ but the profile $\phi$ is still monotonically increasing and positive. By the Sturm theorem for the operator $\mathcal{L}_-(\varepsilon)$ defined by (\ref{L-minus}),  
the spectrum of $\mathcal{L}_-(\varepsilon)$ is positive semi-definite, which follows from $\mathcal{L}_-(\varepsilon) \phi = 0$ with the positive kink profile $\phi$.

Changing the sign of the odd potential reverses the stability conclusion, 
in agreement with the theory.  Figure~\ref{fig:odd-minus} presents the result for $V(x)=-x e^{-x^2}$.  Again $s_0=0$ for the root of $V_{\rm eff}$  in Theorem \ref{th-existence} and we have now $V_{\rm eff}'(0)<0$.  Theorem~\ref{th-splitting} predicts 
$\mu_+ \approx 0.575$ in reasonable agreement with the numerical data 
$\mu_+ \approx 0.527$, again recalling that  $\varepsilon = 1$.  The spectral stability  problem (\ref{spec-stab}) does not have eigenvalues outside the purely imaginary axis, consistent with Remark \ref{rem-stab}. Furthermore,  the component $v = 2 \phi^2 - 1$ is not odd in $x$, in agreement with Theorem \ref{th-parity}. Since the kink profile $\phi$ is still positive, although it is not monotonically increasing, the Sturm theorem still predicts that the spectrum of $\mathcal{L}_-(\varepsilon)$ is positive semi-definite.

\subsection{Kinks in the potential (\ref{potential-2})}

We first set $a=1$ in (\ref{potential-2}) and show the numerical results in Figure \ref{fig:even}. Two symmetric roots $s_0 \approx \pm1.498$ are computed from $V_{\rm eff}$ 
given by (\ref{eff-potential}), in good agreement with the numerical data. 

At the negative root $s_0$, 
we evaluate $V_{\rm eff}'(s_0) \approx 0.135 > 0$ so that 
the theoretical predictions (\ref{eig-L-plus}) and (\ref{eig-spec-stab}) 
give $\mu_+ \approx-0.0540$ 
and $\lambda \approx 0.0270$, in close agreement with the numerical data 
$\mu_+ \approx -0.0540$ and $\lambda \approx 0.0256$. 
At the positive root of $V_{\rm eff}$ we evaluate 
$V_{\rm eff}'(s_0) \approx -0.135<0$ so that 
$\mu_+ \approx 0.054$ in  close agreement with the numerical data 
$\mu_+ \approx 0.0528$. No eigenvalues exist outside the purely imaginary axis 
in (\ref{spec-stab}).

Next, we check how the conclusions change under varying the parameter $a$ 
in the potential (\ref{potential-2}). Figure~\ref{fig:dep-cq-bifurcation} presents the bifurcation diagram for the roots of $V_{\rm eff}$, which can be computed explicitly as 
\begin{equation}
\label{Melnikov}
V_{\rm eff}(s) = 12-a(\pi^2+6s^2)+s[-12+a(\pi^2+4s^2)]\coth(s).
\end{equation}
The two symmetric roots of $V_{\rm eff}$ exist for $a \in (0,a_{\rm crit})$. They coalesce at $s = 0$ in the fold bifurcation at  $a = a_{\rm crit}$, which can be found explicitly as 
\begin{equation}
\label{a-crit}
a=a_{\rm crit}=\frac{12}{\pi^2-6}\approx3.10109.
\end{equation}
Pseudo-arclength continuation of the kink solutions to the stationary equation (\ref{ode}) reproduces both the positive and negative branches 
close to each other  over the entire computed branch, with the slight differences observed in the inset. The difference is explained by the fact that the predictions 
from the effective potential $V_{\rm eff}$ are valid in the asymptotic 
limit $\varepsilon \to 0$, whereas we keep  $\varepsilon=0.1$ in our numerical approximations.  This continuation result naturally extends
the two cases at $a=1$ shown in Figure~\ref{fig:even}.

\begin{figure}[htbp]
	\centering
	\includegraphics[width=0.475\linewidth,height=0.4\textwidth]{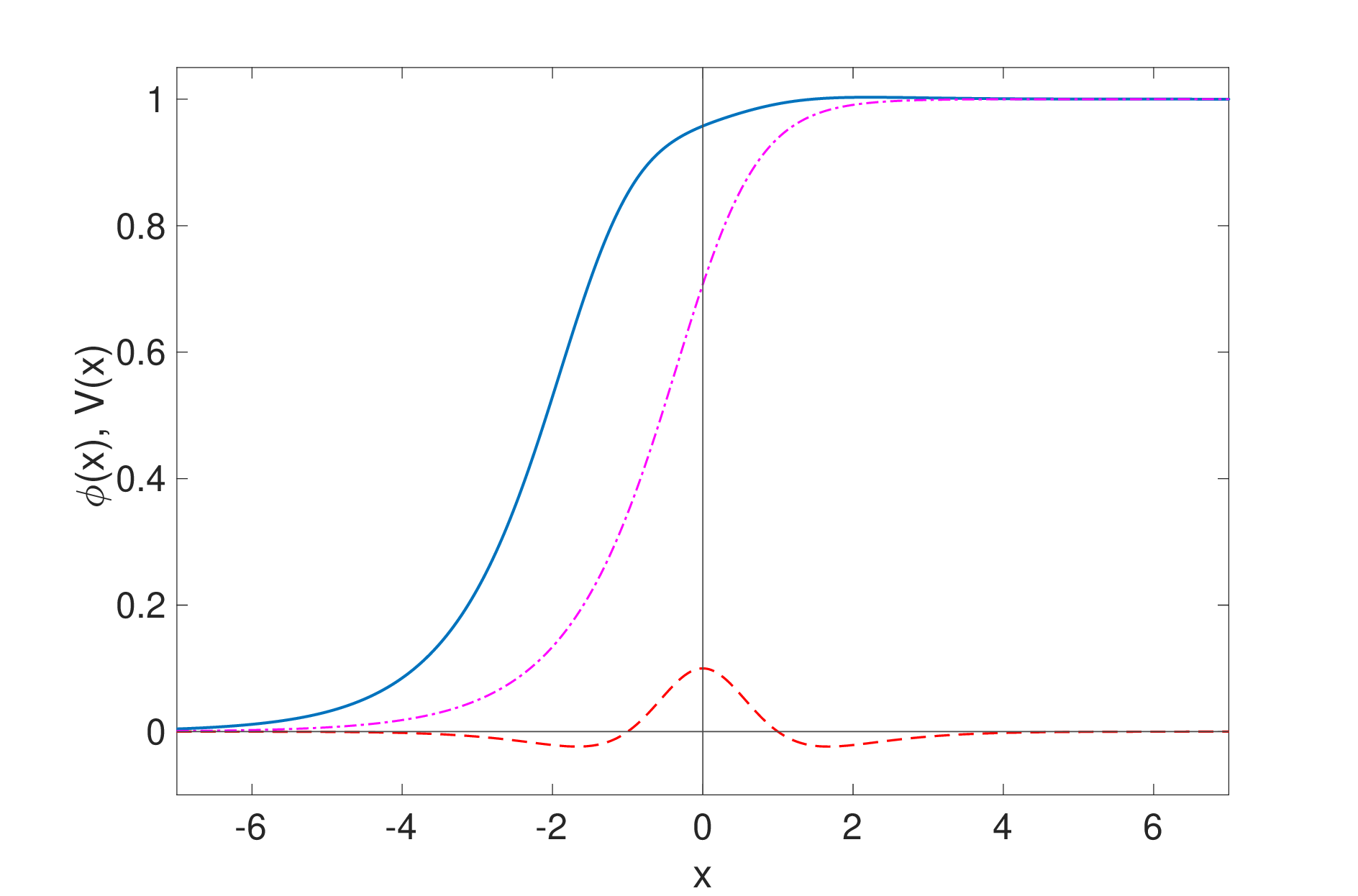}
	\includegraphics[width=0.475\linewidth,height=0.4\textwidth]{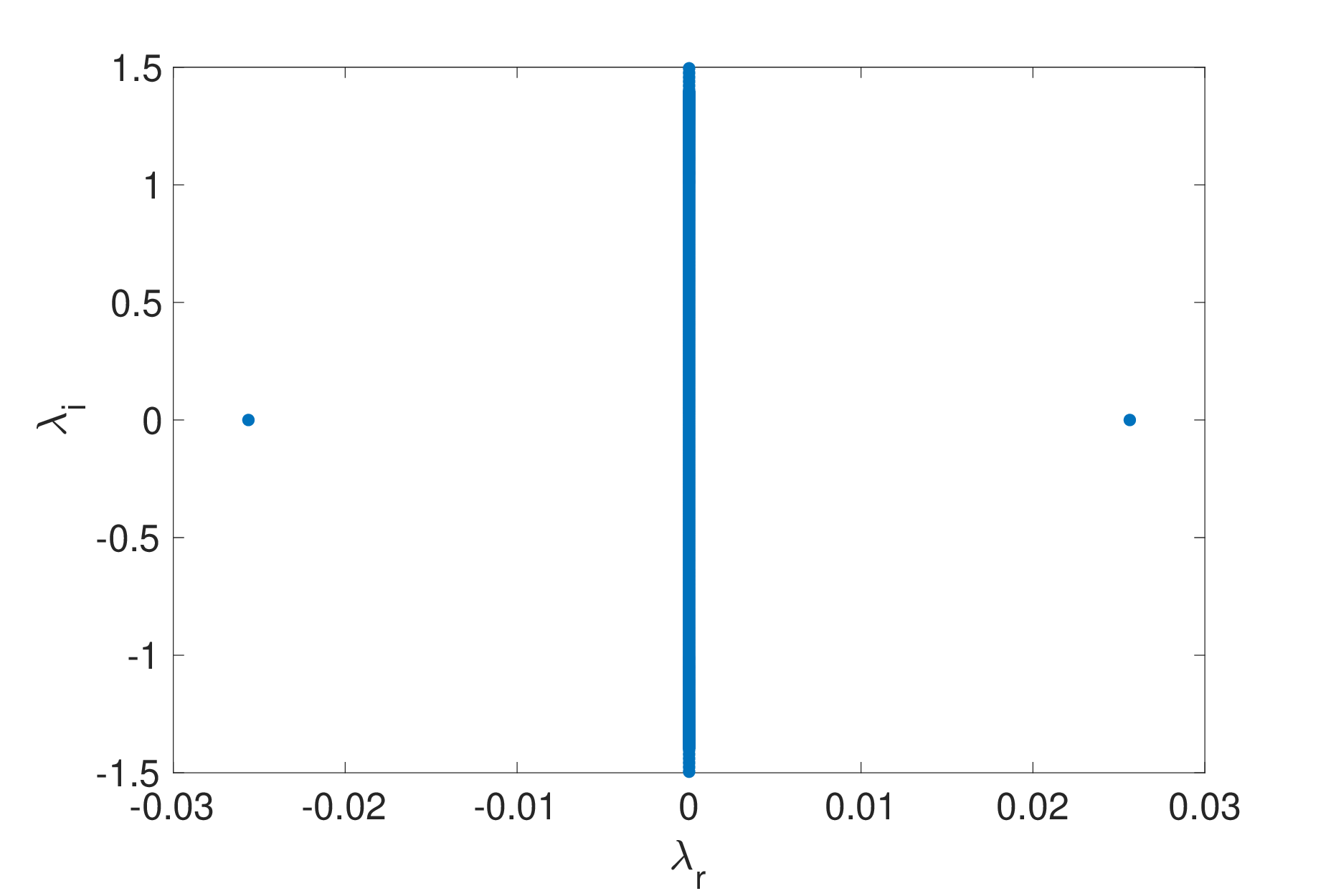}
	\includegraphics[width=0.475\linewidth,height=0.4\textwidth]{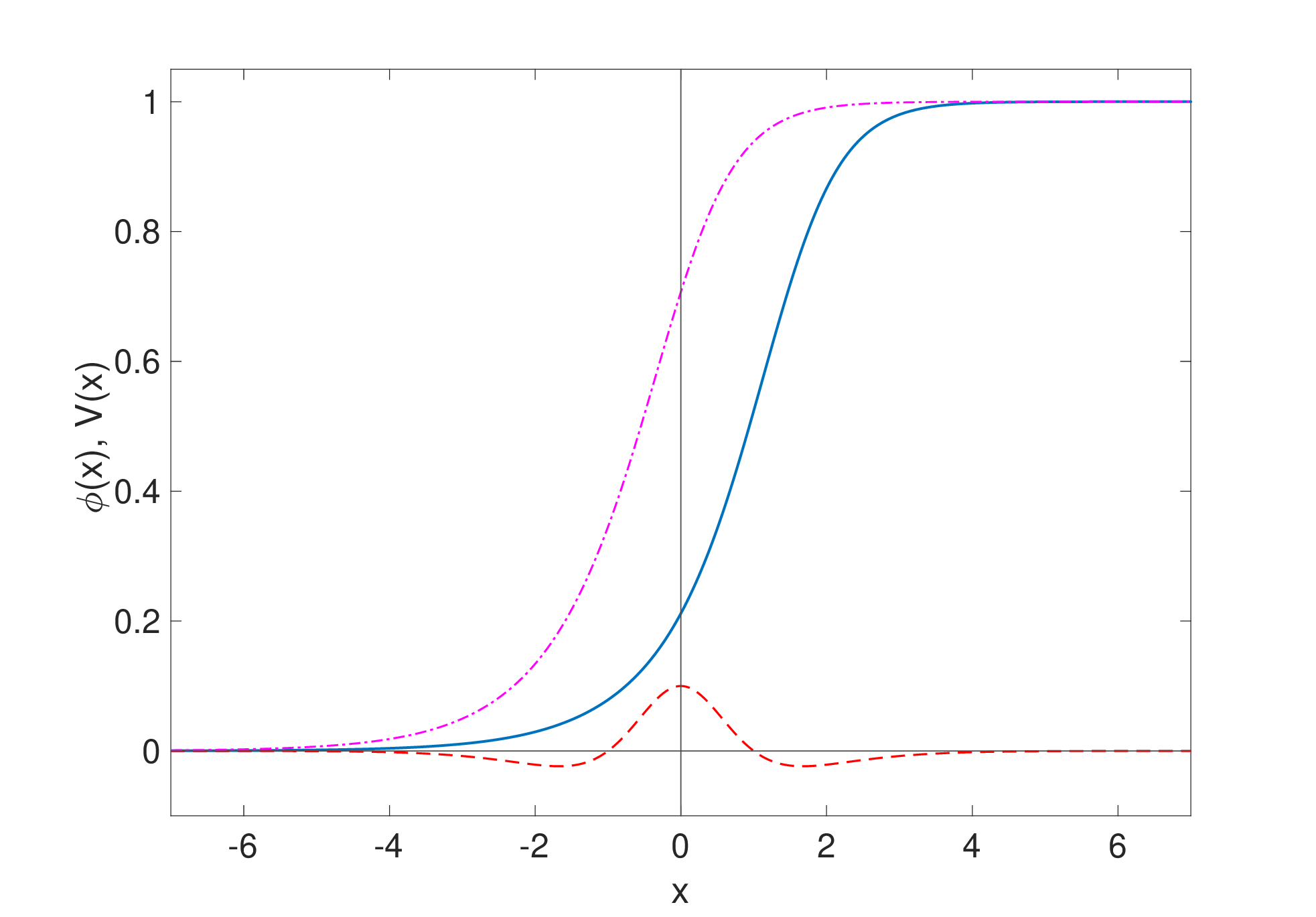}
	\includegraphics[width=0.475\linewidth,height=0.4\textwidth]{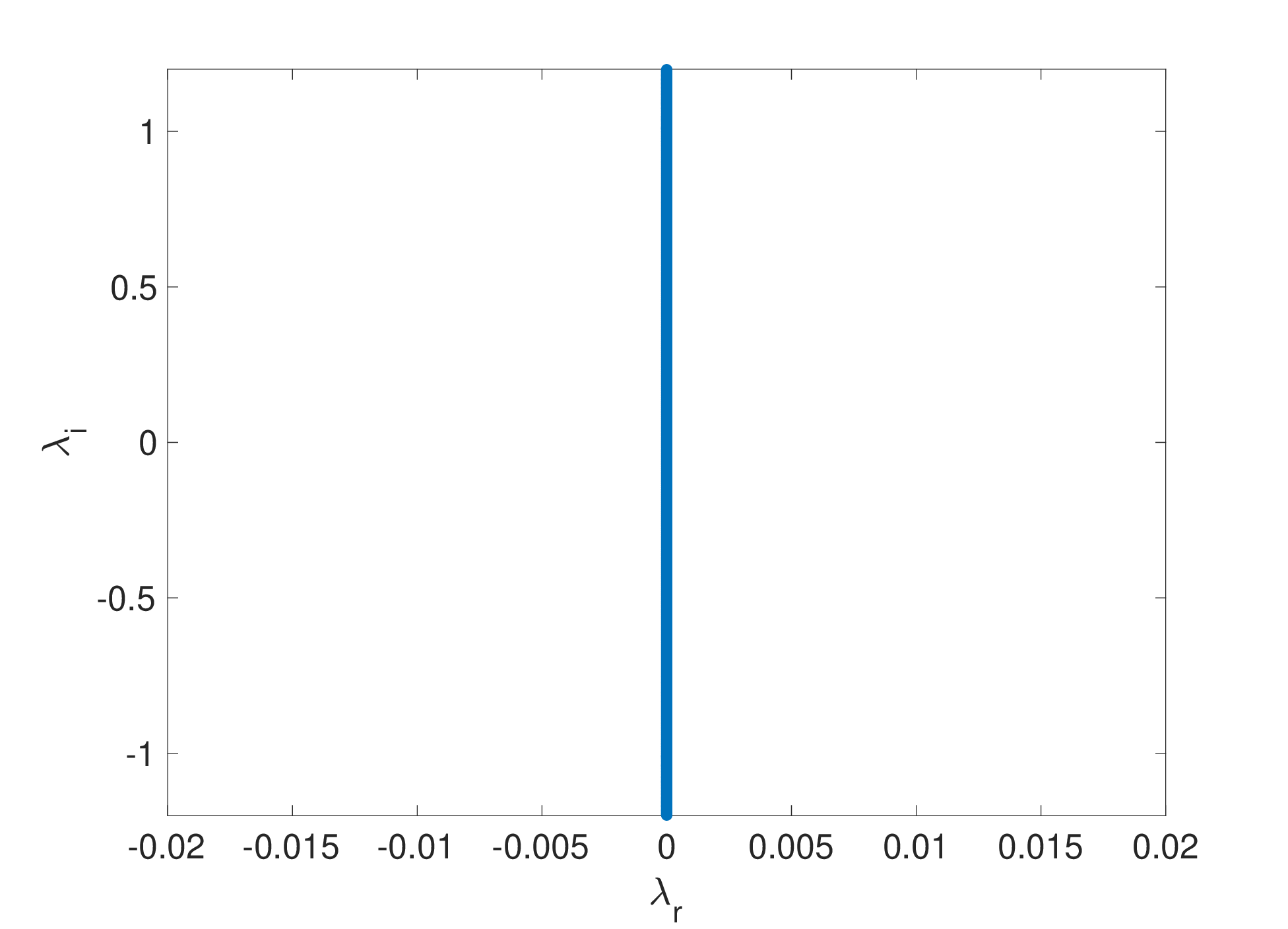}
	\caption{The same as Figure \ref{fig:odd-plus} but for the potential $V(x) = (1-x^2){\rm sech}^2(x)$ with $\varepsilon = 0.1$. 
		Top row: the negative root of $V_{\rm eff}$. Bottom row: the positive root of $V_{\rm eff}$. }
	\label{fig:even}
\end{figure}

\begin{figure}[htbp]
	\centering
	\includegraphics[width=0.7\textwidth]{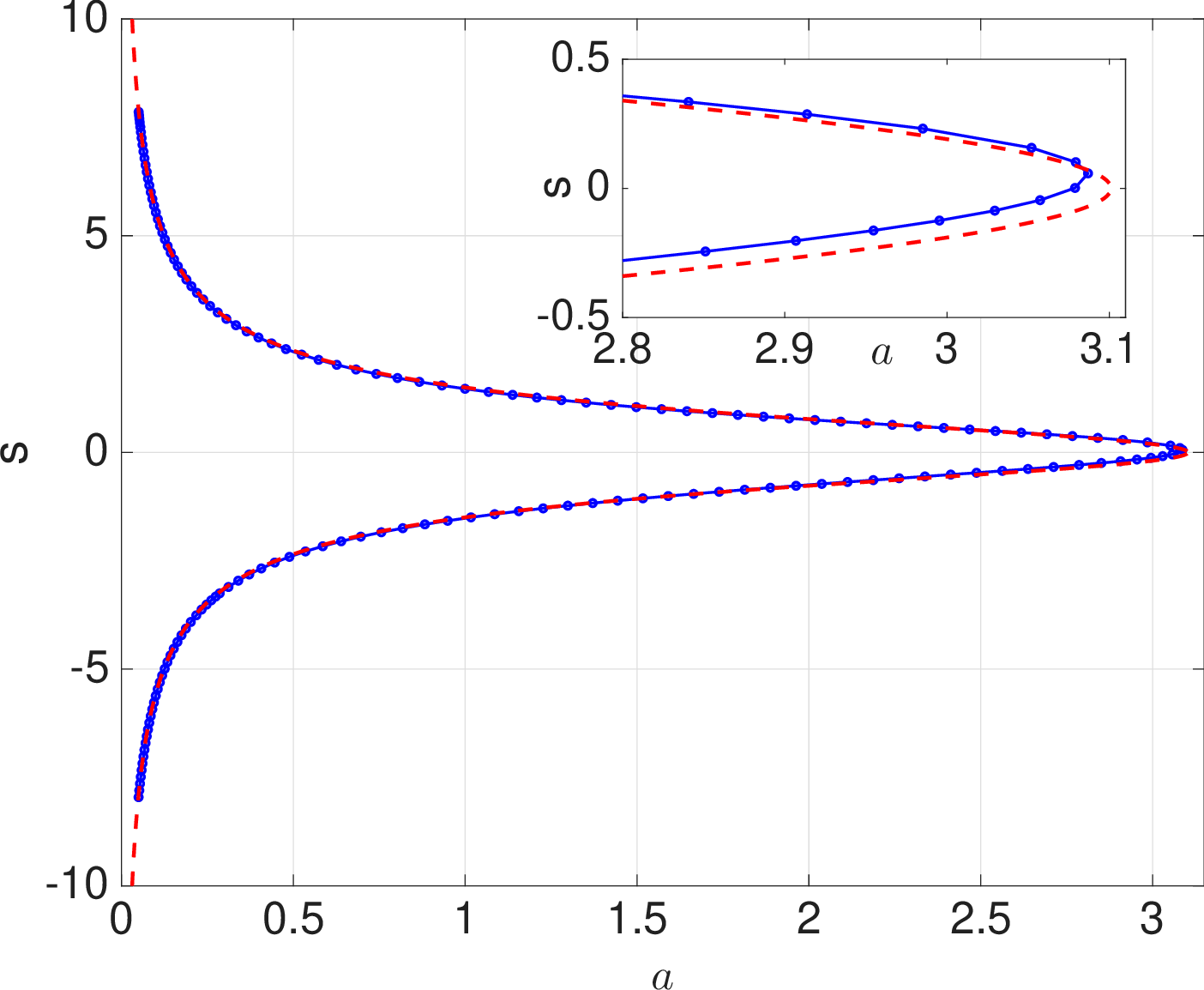}
	\caption{Bifurcation diagram for the potential $V(x) = (1-a x^2){\rm sech}^2(x)$ for $\varepsilon = 0.1$, showing the numerically computed kink-center location $s$ from pseudo-arclength continuation together with the roots of $V_{\rm eff}$ given by (\ref{Melnikov}).  The diagram captures the positive and negative branches and their merger near the fold bifurcation at $a_{\rm crit}$ given by (\ref{a-crit}).}
	\label{fig:dep-cq-bifurcation}
\end{figure}

Figure \ref{fig:dep-cq-spectral-comparison} compares the numerically computed eigenvalues $\mu_+(\varepsilon)$ of $\mathcal{L}_+(\varepsilon)$ (left) and $\lambda(\varepsilon)$ of the spectral stability problem (\ref{spec-stab}) (right) with the leading-order asymptotic approximations 
(\ref{eig-L-plus}) and (\ref{eig-spec-stab}) for the same value 
of $\varepsilon = 0.1$ and varying values of $a$. 

The eigenvalue $\mu_+(\varepsilon)$ is negative for the negative root $s_0$ of $V_{\rm eff}$ and positive for the positive root $s_0$. The agreement is very good and consistent with the approximation in Figure \ref{fig:dep-cq-spectral-comparison}. The discrepancy between the two curves is visible for $a>1.5$ since the value of $\varepsilon$ is fixed. The eigenvalue $\mu_{+}(\varepsilon)$ vanishes as $a \to a_{\rm crit}$ due to the fold bifurcation. 

The eigenvalue $\lambda(\varepsilon)$ is positive only for the negative root $s_0$ of $V_{\rm eff}$. To obtain an accurate numerical approximation 
of $\lambda(\varepsilon)$ when it is small, we extend the computational domain to resolve the slowly decaying eigenfunctions. The agreement is again very good and consistent with the approximation in Figure \ref{fig:dep-cq-spectral-comparison}. The eigenvalue $\lambda(\varepsilon)$ also vanishes as $a \to a_{\rm crit}$. Near points where $\lambda(\varepsilon)$ is very small, relative errors are naturally sensitive, so the absolute comparison is the more meaningful diagnostic.

Taken together, Figures ~\ref{fig:dep-cq-bifurcation} and \ref{fig:dep-cq-spectral-comparison} provide a
direct numerical verification of the pinning criterion established in Theorem \ref{th-existence}, the eigenvalue-splitting mechanism of Theorem~\ref{th-splitting}, and the unstable-eigenvalue asymptotics of Theorem~\ref{th-stability}. More specifically, the numerical results therefore support the principal conclusions of the analysis: the effective potential $V_{\rm eff}$ correctly selects the pinned kink centers near a root $s_0$, the sign of $V_{\rm eff}'(s_0)$ determines the sign of the small eigenvalue $\mu_+(\varepsilon)$ of $\mathcal{L}_+(\varepsilon)$, and on the branch with $V_{\rm eff}'(s_0)>0$ the resulting negative direction of $\mathcal{L}_+(\varepsilon)$ produces the real positive unstable eigenvalue $\lambda(\varepsilon)$ of the spectral stability problem (\ref{spec-stab}).

\begin{figure}[htbp]
	\centering
	\begin{minipage}[t]{0.48\textwidth}
		\centering
		\includegraphics[width=\linewidth,height=0.75\textwidth]{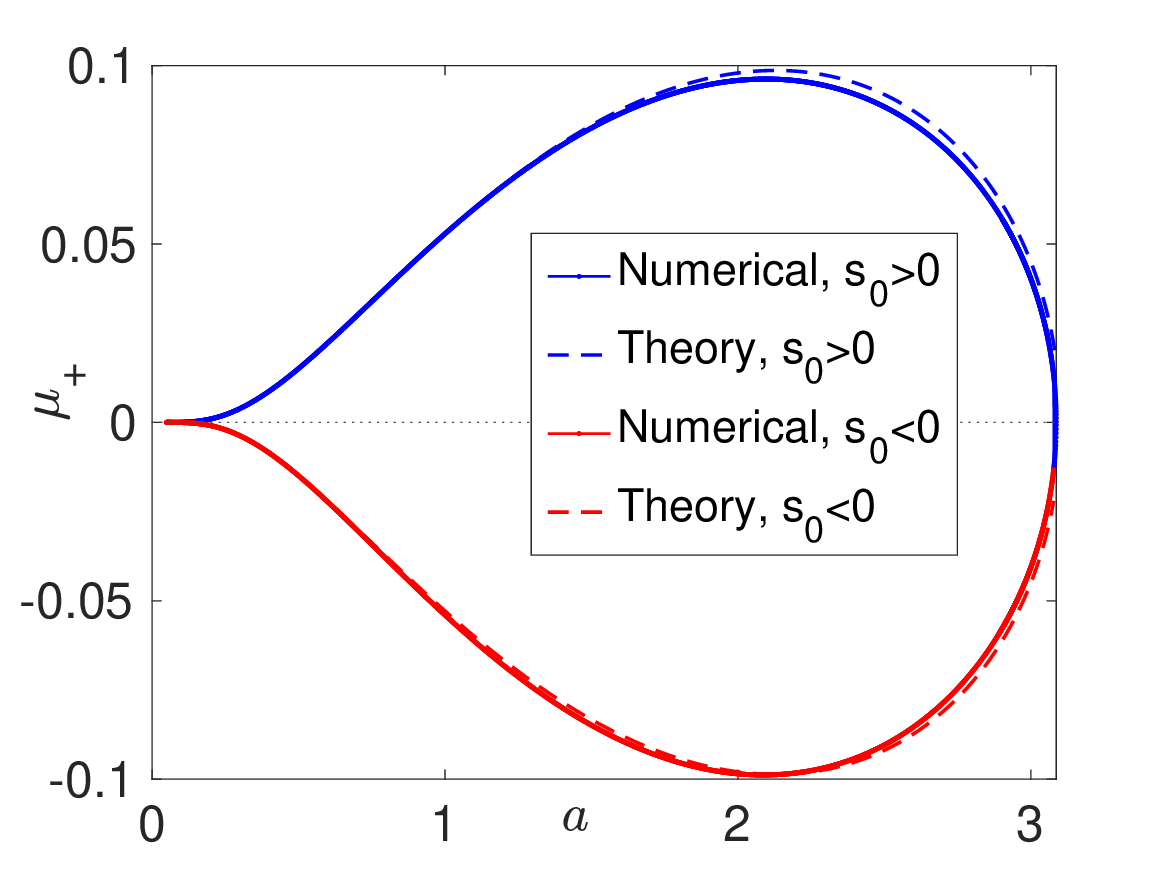}
		\centerline{\small }
	\end{minipage}\hfill
	\begin{minipage}[t]{0.48\textwidth}
		\centering
		\includegraphics[width=\linewidth,height=0.75\textwidth]{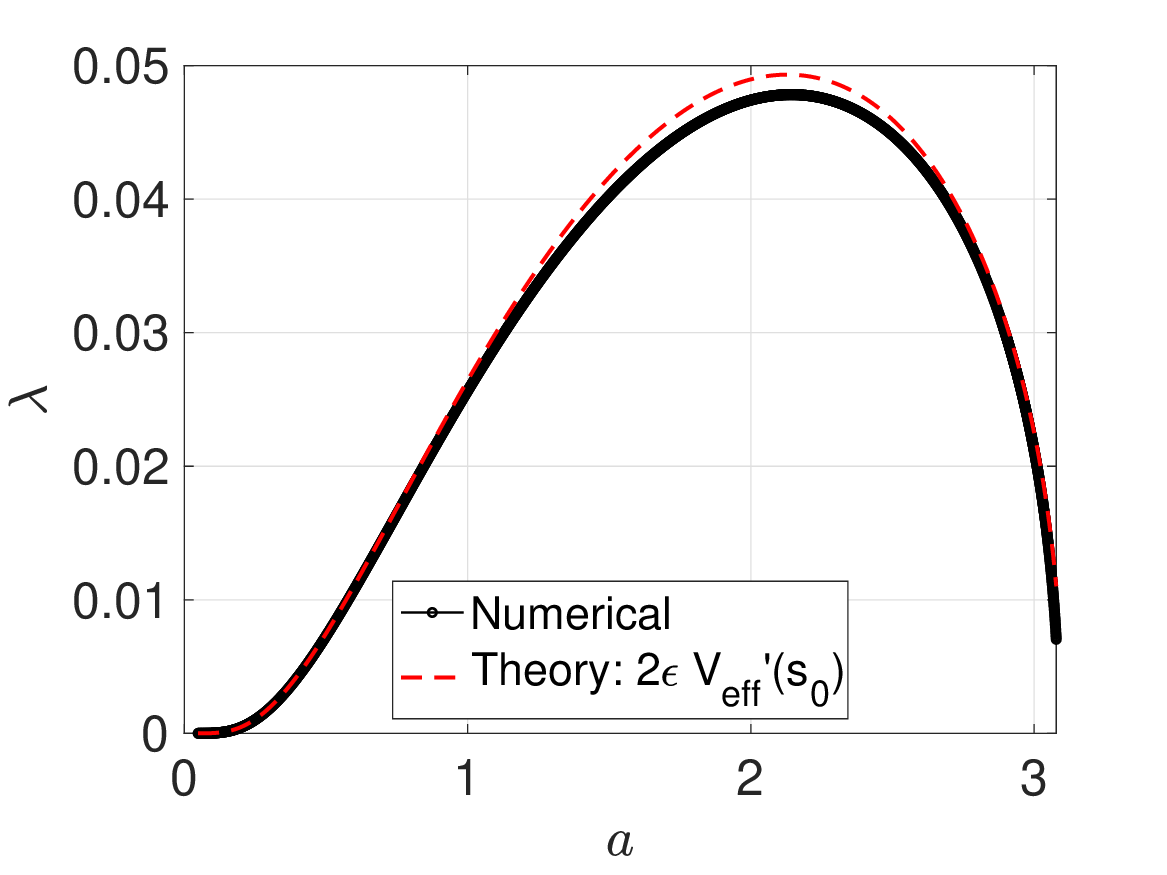}
		\centerline{\small }
	\end{minipage}
	\caption{Comparison of the numerically detected eigenvalues $\mu_+(\varepsilon)$ of $\mathcal{L}_+(\varepsilon)$ (left) 
		and $\lambda(\varepsilon)$ of the spectral stability problem (\ref{spec-stab}) with the leading-order theoretical predictions (\ref{eig-L-plus}) and (\ref{eig-spec-stab}).}
	\label{fig:dep-cq-spectral-comparison}
\end{figure}

\subsection{Dynamics in the potential (\ref{potential-2})}

Figure~\ref{fig:dep-cq-dynamics} presents the dynamical evolution of the perturbed unstable kink associated with the negative root of $V_{\rm eff}$. The initial condition consists of the kink shown on the top left
panel of  Figure~\ref{fig:even} perturbed via a (small) multiple of the unstable eigenvector of the spectral stability problem (\ref{spec-stab}).
In agreement with the unstable real positive eigenvalue $\lambda$ 
on the top right panel of Figure~\ref{fig:even}, the perturbations grow exponentially along the unstable
eigendirection during the early stage of the evolution.
In the time-dependent simulation, this results in the
departure of the kink from the neighborhood of the unstable pinned kink 
position. In the process, it emits extended wave radiation that escapes
toward the boundary. Subsequently, the kink
moves to the positive side where it approaches the stable pinned kink monotonically.

This dynamical evolution is also consistent with predictions of the viscous particle equation (\ref{viscous-particle}) in Remark \ref{rem-particle}.
The negative root $s_0$ of $V_{\rm eff}$, for which $V_{\rm eff}'(s_0) > 0$, is a source, whereas the positive root $s_0$ of $V_{\rm eff}$, for which $V_{\rm eff}'(s_0)<0$, is a sink for $\varepsilon > 0$. The trajectory of the particle 
from (\ref{viscous-particle}) is shown in Figure \ref{fig:dep-cq-dynamics} by the dashed line. We observe a satisfactory agreement between the monotone dynamics of the kink observed at the PDE level and the corresponding one-dimensional viscous particle dynamics at the ODE level.

We point out that the dynamics of black solitons in the NLS equation with an external potential is essentially two-dimensional with either saddle point or spiral source dynamics, e.g., see \cite{KevPel08}. The difference is that the splitting of eigenvalues in the spectral stability problem 
(\ref{spec-stab}) features the square root splitting in $\varepsilon$, 
in contrast to the expansion (\ref{eig-spec-stab}).

It is interesting to point out (data not shown here) that in the
case of performing the computation in a finite domain, the radiation may
eventually be reflected from the boundaries (e.g., in the case
of Dirichlet or Neumann boundary conditions)
and impinge on the kink, leading it to bounce back
to the negative side which, in turn, will lead to a reemission event
and a repeat of the transition toward the stable pinned kink and so on. Nevertheless, this feature due to the finite computational domain 
is outside the scope of the present work and hence
we do not explore it further.

\begin{figure}[htbp]
	\centering
	\includegraphics[width=0.7\textwidth]{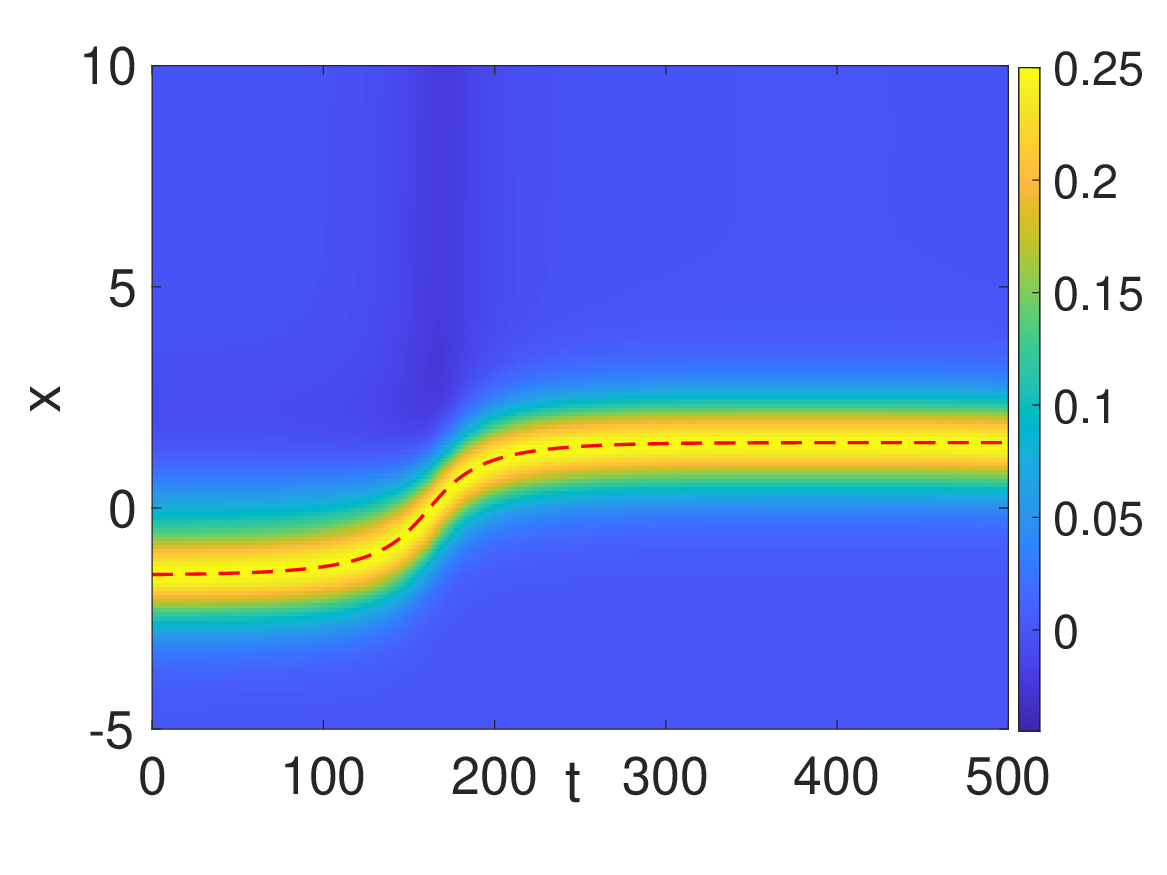}
	\caption{Dynamical evolution of the unstable kink for the potential $V(x) = (1-x^2){\rm sech}^2(x)$ with $\varepsilon = 0.1$.  The kink initially located near $s_0 \approx-1.5$ becomes exponentially unstable and emits a substantial wavepacket of radiation around $t=200$, after which it moves to the positive side and settles in the vicinity of the stable root $s_0 \approx 1.5$. The red dashed line shows the dynamical evolution of the viscous particle equation (\ref{viscous-particle}). The spatio-temporal
		contour plot shows the field $|\psi|^2(1-|\psi|^2)$ which has the form of a
		single-humped pulse.}
	\label{fig:dep-cq-dynamics}
\end{figure}

\bibliographystyle{amsplain}
\bibliography{KinkPersistence_refs}

\end{document}